\documentclass[draftcls,onecolumn]{IEEEtran}
\usepackage{amsmath}
\usepackage{amssymb}
\usepackage{amsfonts}
\usepackage{slashbox}
\usepackage{graphicx}
\usepackage{cite, amsmath, amsfonts, amssymb, psfrag, epsfig,tikz}
\usetikzlibrary{shapes,shadows,arrows}
\newtheorem{proposition}{{Proposition}}
\newtheorem{definition}{{Definition}}
\newtheorem{theorem}{{Theorem}}

\newtheorem{corollary}{{Corollary}}

\DeclareMathAlphabet{\mathpzc}{OT1}{pzc}{m}{it}

\begin{document}

\title{On the Reliability Function of Variable-Rate Slepian-Wolf Coding}


\author{Jun Chen, Da-ke He, Ashish Jagmohan, Luis A. Lastras-Monta\~{n}o
}

\maketitle

\begin{abstract}
The reliability function of variable-rate Slepian-Wolf coding is
linked to the reliability function of channel coding with constant
composition codes, through which computable lower and upper bounds
are derived. The bounds coincide at rates close to the
Slepian-Wolf limit, yielding a complete characterization of the
reliability function in that rate regime. It is shown that
variable-rate Slepian-Wolf codes can significantly outperform
fixed-rate Slepian-Wolf codes in terms of rate-error tradeoff. The
reliability function of variable-rate Slepian-Wolf coding with
rate below the Slepian-Wolf limit is determined. In sharp contrast
with fixed-rate Slepian-Wolf codes for which the correct decoding
probability decays to zero exponentially fast if the rate is below
the Slepian-Wolf limit, the correct decoding probability of
variable-rate Slepian-Wolf codes can be bounded away from zero.
\end{abstract}


\section{Introduction}

%study whether the reliability function can be determined when
%$R=H(P_X)$. overflow constraint

Consider the problem (see Fig. 1) of compressing
$X^n=(X_1,X_2,\cdots,X_n)$ with side information
$Y^n=(Y_1,Y_2,\cdots,Y_n)$ available only at the decoder. Here
$\{(X_i,Y_i)\}_{i=1}^\infty$ is a joint memoryless source with
zero-order joint probability distribution $P_{XY}$ on finite
alphabet $\mathcal{X}\times\mathcal{Y}$. Let $P_X$ and $P_Y$ be
the marginal probability distributions of $X$ and $Y$ induced by
the joint probability distribution $P_{XY}$. Without loss of
generality, we shall assume $P_X(x)>0, P_Y(y)>0$ for all
$x\in\mathcal{X},y\in\mathcal{Y}$. This problem was first studied
by Slepian and Wolf in their landmark paper \cite{SW73}. They
proved a surprising result that the minimum rate for
reconstructing $X^n$ at the decoder with asymptotically zero error
probability (as block length $n$ goes to infinity) is $H(X|Y)$,
which is the same as the case where the side information $Y^n$ is
also available at the encoder. The fundamental limit $H(X|Y)$ is
often referred to as the Slepian-Wolf limit. We shall assume
$H(X|Y)>0$ throughout this paper.

\begin{figure}[hbt]
\centering
\begin{psfrags}
\psfrag{x}[c]{$X^n$}%
\psfrag{y}[c]{$Y^n$}%
\psfrag{r1}[c]{$R$}%
\psfrag{xhat}[c]{$\widehat{X}^n$}%
\psfrag{yhat}[c]{$\widehat{Y}^n$}%
\psfrag{en1}[c]{Encoder}%
\psfrag{de}[c]{Decoder}%
\includegraphics[scale=1.2]{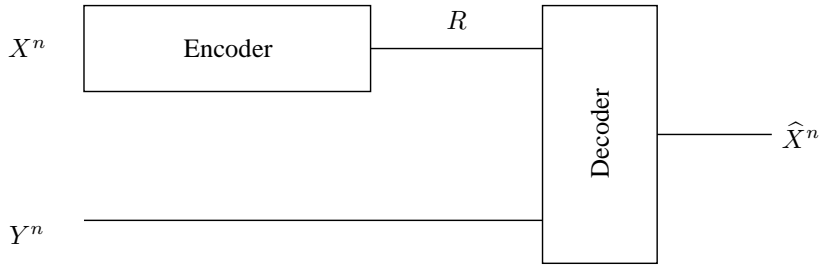}
\caption{Slepian-Wolf coding}
\end{psfrags}
\end{figure}

Different from conventional lossless source coding, where most
effort has been devoted to variable-rate coding schemes, research
on Slepian-Wolf coding has almost exclusively focused on
fixed-rate codes. This phenomenon can be partly explained by the
influence of channel coding. It is well known that there is an
intimate connection between channel coding and Slepian-Wolf
coding. Intuitively, one may view $Y^n$ as the channel output
generated by channel input $X^n$ through discrete memoryless
channel $P_{Y|X}$, where $P_{Y|X}$ is the probability transition
matrix from $\mathcal{X}$ to $\mathcal{Y}$ induced by the joint
probability probability distribution $P_{XY}$. Since $Y^n$ is not
available at the encoder, Slepian-Wolf coding is, in a certain
sense, similar to channel coding without feedback. In a channel
coding system, there is little incentive to use variable-rate
coding schemes if no feedback link exists from the receiver to the
transmitter. Therefore, it seems justifiable to focus on
fixed-rate codes in Slepian-Wolf coding.

This viewpoint turns out to be misleading. We shall show that
variable-rate Slepian-Wolf codes can significantly outperform
fixed-rate codes in terms of rate-error tradeoff. Specifically, it
is revealed that variable-rate Slepian-Wolf codes can beat the
sphere-packing bound for fixed-rate Slepian-Wolf codes at rates
close to the Slepian-Wolf limit\footnote{Note that the same
conclusion is trivially true if the rate is greater than $H(X)$
since in this case one can achieve zero error probability using
variable-rate coding schemes.}. It is known that the correct
decoding probability of fixed-rate Slepian-Wolf codes decays to
zero exponentially fast if the rate is below the Slepian-Wolf
limit. Somewhat surprisingly, the decoding error probability of
variable-rate Slepian-Wolf codes can be bounded away from one even
when they are operated below the Slepian-Wolf limit, and the
performance degrades graciously as the rate goes to zero.
Therefore, variable-rate Slepian-Wolf coding is considerably more
robust.

The rest of this paper is organized as follows. In Section
\ref{fixedrate}, we review the existing bounds on the reliability
function of fixed-rate Slepian-Wolf coding, and point out the
intimate connections with their counterparts in channel coding. In
Section \ref{variablerate}, we characterize the reliability
function of variable-rate Slepian-Wolf coding by leveraging the
reliability function of channel coding with constant composition
codes. Computable lower and upper bounds are derived. The bounds
coincide at rates close to the Slepian-Wolf limit. The correct
decoding probability of variable-rate Slepian-Wolf coding with
rate below the Slepian-Wolf limit is studied in Section
\ref{below}. An illustrative example is given in Section
\ref{example}. We conclude the paper in Section \ref{conclusion}.
Throughout this paper, we assume the logarithm function is to base
$e$ unless specified otherwise.

%For fixed-rate Slepian-Wolf codes, one needs the error exponent
%has the uniform convergence property (uniform with respect to
%different types). For variable-rate codes, one does not need such
%an assumption since only local types have to be considered, all
%other type classes can be compressed losslessly.

\section{Fixed-Rate Slepian-Wolf Coding and Channel Coding}\label{fixedrate}

To facilitate the comparisons between the performances of
fixed-rate Slepian-Wolf coding and variable-rate coding, we shall
briefly review the existing bounds on the reliability function of
fixed-rate Slepian-Wolf coding. It turns out that a most
instructive way is to first consider their counterparts in channel
coding. The reason is two-fold. First, it provides the setup to
introduce several important definitions. Second and more
important, it will be clear that the reliability function of
fixed-rate Slepian-Wolf coding is closely related to that of
channel coding; indeed, such a connection will be further explored
in the context of variable-rate Slepian-Wolf coding.

For any probability distributions $P, Q$ on $\mathcal{X}$ and
probability transition matrices
$V,W:\mathcal{X}\rightarrow\mathcal{Y}$, we use $H(P)$, $I(P,V)$,
$D(Q\|P)$, and $D(W\|V|P)$ to denote the standard entropy, mutual
information, divergence, and conditional divergence functions;
specifically, we have
\begin{eqnarray*}
&&H(P)=-\sum\limits_{x}P(x)\log P(x),\\
&&I(P,V)=\sum\limits_{x,y}P(x)V(y|x)\log\frac{V(y|x)}{\sum_{x'}P(x')V(y|x')},\\
&&D(Q\|P)=\sum\limits_{x}Q(x)\log\frac{Q(x)}{P(x)},\\
&&D(W\|V|P)=\sum\limits_{x,y}P(x)W(y|x)\log\frac{W(y|x)}{V(y|x)}.
\end{eqnarray*}

The main technical tool we need is the method of types. First, we
shall quote a few basic definitions from \cite{CK81B}. Let
$\mathcal{P}(\mathcal{X})$ denote the set of all probability
distributions on $\mathcal{X}$. The type of a sequence
$x^n\in\mathcal{X}^n$, denoted as $P_{x^n}$, is the empirical
probability distribution of $x^n$. Let
$\mathcal{P}_n(\mathcal{X})$ denote the set consisting of the
possible types of sequences $x^n\in\mathcal{X}^n$. For any
$P\in\mathcal{P}_n(\mathcal{X})$, the type class
$\mathcal{T}_n(P)$ is the set of sequences in $\mathcal{X}^n$ of
type $P$. We will make frequent use of the following elementary
results:
\begin{eqnarray}
&&|\mathcal{P}_n(\mathcal{X})|\leq (n+1)^{|\mathcal{X}|},\label{cardinality}\\
&&\frac{1}{(n+1)^{|\mathcal{X}|}}e^{nH(P)}\leq
|\mathcal{T}_n(P)|\leq e^{nH(P)},\quad
P\in\mathcal{P}_n(\mathcal{X}),\label{typeclass}\\
&&\prod\limits_{i=1}^nP(x_i)=e^{-n\left[D(Q\|P)+H(Q)\right]},\quad
x^n\in\mathcal{T}_n(Q),
Q\in\mathcal{P}_n(\mathcal{X}),P\in\mathcal{P}(\mathcal{X}).\label{typeprobability}
\end{eqnarray}

A block code\footnote{More precisely, a block code $\mathcal{C}_n$
is an ordered collection of sequences in $\mathcal{X}^n$. We allow
$\mathcal{C}_n$ to contain identical sequences. Moreover, for any
set $\mathcal{A}\subseteq\mathcal{X}^n$, we say
$\mathcal{C}_n\subseteq\mathcal{A}_n$ if $x^n\in\mathcal{A}$ for
all $x^n\in\mathcal{C}_n$. Note that
$\mathcal{C}_n\subseteq\mathcal{A}$ does not imply
$|\mathcal{C}_n|\leq|\mathcal{A}|$.} $\mathcal{C}_n$ is a set of
sequences in $\mathcal{X}^n$. The rate of $\mathcal{C}_n$ is
defined as
\begin{eqnarray*}
R(\mathcal{C}_n)=\frac{1}{n}\log|\mathcal{C}^n|.
\end{eqnarray*}
Given a channel $W_{Y|X}:\mathcal{X}\rightarrow\mathcal{Y}$, a
block code $\mathcal{C}_n\subseteq\mathcal{X}^n$, and channel
output $Y^n\in\mathcal{Y}^n$, the output of the optimal
\textit{maximum likelihood} (ML) decoder is
\begin{eqnarray*}
\widehat{X}^n=\arg\min\limits_{x^n\in\mathcal{C}_n}-\sum\limits_{i=1}^n\log
W_{Y|X}(Y_i|x_i),
\end{eqnarray*}
where the ties are broken in an arbitrary manner. The average
decoding error probability of block code $\mathcal{C}_n$ over
channel $W_{Y|X}$ is defined as
\begin{eqnarray*}
P_e(\mathcal{C}_n,W_{Y|X})=\frac{1}{|\mathcal{C}_n|}\sum\limits_{x^n\in\mathcal{C}_n}\mbox{Pr}\{\widehat{X}^n\neq
x^n | x^n\mbox{  is transmitted}\}.
\end{eqnarray*}
The maximum decoding error probability of block code
$\mathcal{C}_n$ over channel $W_{Y|X}$ is defined as
\begin{eqnarray*}
P_{e,\max}(\mathcal{C}_n,W_{Y|X})=\max\limits_{x^n\in\mathcal{C}_n}\mbox{Pr}\{\widehat{X}^n\neq
x^n | x^n\mbox{  is transmitted}\}.
\end{eqnarray*}
The average correct decoding probability of block code
$\mathcal{C}_n$ over channel $W_{Y|X}$ is defined as
\begin{eqnarray*}
P_c(\mathcal{C}_n,W_{Y|X})=1-P_e(\mathcal{C}_n,W_{Y|X}).
\end{eqnarray*}

\begin{definition}\label{def:blockcode}
Given a channel $W_{Y|X}:\mathcal{X}\rightarrow\mathcal{Y}$, we
say an error exponent $E\geq 0$ is achievable with block codes at
rate $R$ if for any $\delta>0$, there exists a sequence of block
codes codes $\{\mathcal{C}_n\}$ such that
\begin{eqnarray}
&&\liminf\limits_{n\rightarrow\infty}R(\mathcal{C}_n)\geq R-\delta,\nonumber\\
&&\limsup\limits_{n\rightarrow\infty}-\frac{1}{n}\log
P_e(\mathcal{C}_n,W_{Y|X})\geq E-\delta.\label{averageerror}
\end{eqnarray}
The largest achievable error exponent at rate $R$ is denoted by
$E(W_{Y|X},R)$. The function $E(W_{Y|X},\cdot)$ is referred to as
the reliability function of channel $W_{Y|X}$. Similarly, we say a
correct decoding exponent $E^c\geq 0$ is achievable with block
channel codes at rate $R$ if for any $\delta>0$, there exists a
sequence of block codes $\{\mathcal{C}_n\}$ such that
\begin{eqnarray*}
&&\liminf\limits_{n\rightarrow\infty}R(\mathcal{C}_n)\geq R-\delta,\\
&&\liminf\limits_{n\rightarrow\infty}-\frac{1}{n}\log
P_c(\mathcal{C}_n,W_{Y|X})\leq E^c+\delta.
\end{eqnarray*}
The smallest achievable correct decoding exponent at rate $R$ is
denoted by $E^c(W_{Y|X},R)$. It will be seen that $E^c(W_{Y|X},R)$
is positive if and only if $R>C(W_{Y|X})$, where
$C(W_{Y|X})\triangleq\max_{Q_X}I(Q_X,W_{Y|X})$ is the capacity of
channel $W_{Y|X}$. Therefore, we shall refer to the function
$E^c(W_{Y|X},\cdot)$ as the reliability function of channel
$W_{Y|X}$ above the capacity.
\end{definition}
Remark: Given any block code $\mathcal{C}_n$ of average decoding
error probability $P_e(\mathcal{C}_n,W_{Y|X})$, we can expurgate
the worst half of the codewords so that the maximum decoding error
probability of the resulting code is bounded above by
$2P_e(\mathcal{C}_n,W_{Y|X})$. Therefore, the reliability function
$E(W_{Y|X},\cdot)$ is unaffected if we replace
$P_e(\mathcal{C}_n,W_{Y|X})$ by
$P_{e,\max}(\mathcal{C}_n,W_{Y|X})$ in (\ref{averageerror}).

\begin{definition}\label{def:constantcomposition}
Given a probability distribution $Q_X\in\mathcal{P}(\mathcal{X})$
and a channel $W_{Y|X}:\mathcal{X}\rightarrow\mathcal{Y}$, we say
an error exponent $E\geq 0$ is achievable at rate $R$ with
constant composition codes of type approximately $Q_X$ if for any
$\delta>0$, there exists a sequence of block codes codes
$\{\mathcal{C}_n\}$ with
$\mathcal{C}_n\subseteq\mathcal{T}_n(P_n)$ for some
$P_n\in\mathcal{P}_n(\mathcal{X})$ such that
\begin{eqnarray*}
&&\lim\limits_{n\rightarrow\infty}\|P_n-Q_X\|=0,\\
&&\liminf\limits_{n\rightarrow\infty}R(\mathcal{C}_n)\geq R-\delta,\\
&&\limsup\limits_{n\rightarrow\infty}-\frac{1}{n}\log
P_e(\mathcal{C}_n,W_{Y|X})\geq E-\delta,
\end{eqnarray*}
where $\|\cdot\|$ is the $l_1$ norm. The largest achievable error
exponent at rate $R$ for constant composition codes of type
approximately $Q_X$  is denoted by $E(Q_X,W_{Y|X},R)$. The
function $E(Q_X,W_{Y|X},\cdot)$ is referred to as the reliability
function of channel $W_{Y|X}$ for constant composition codes of
type approximately $Q_X$. Similarly, we say a correct decoding
exponent $E^c\geq 0$ is achievable at rate $R$ with constant
composition codes of type approximately $Q_X$  if for any
$\delta>0$, there exists a sequence of block codes
$\{\mathcal{C}_n\}$ with
$\mathcal{C}_n\subseteq\mathcal{T}_n(P_n)$ for some
$P_n\in\mathcal{P}_n(\mathcal{X})$ such that
\begin{eqnarray}
&&\lim\limits_{n\rightarrow\infty}\|P_n-Q_X\|=0,\nonumber\\
&&\liminf\limits_{n\rightarrow\infty}R(\mathcal{C}_n)\geq R-\delta,\nonumber\\
&&\liminf\limits_{n\rightarrow\infty}-\frac{1}{n}\log
P_c(\mathcal{C}_n,W_{Y|X})\leq E^c+\delta.\label{average2}
\end{eqnarray}
The smallest achievable correct decoding exponent at rate $R$ for
constant composition codes of type approximately $Q_X$ is denoted
by $E^c(Q_X,W_{Y|X},R)$.
\end{definition}
Remark: The reliability function $E(Q_X,W_{Y|X},\cdot)$ is
unaffected if we replace $P_e(\mathcal{C}_n,W_{Y|X})$ by
$P_{e,\max}(\mathcal{C}_n,W_{Y|X})$ in (\ref{average2}).

Let $|t|^+=\max\{0,t\}$ and
$d_{W_{Y|X}}(x,\widetilde{x})=-\log\sum_{y}\sqrt{W_{Y|X}(y|x)W_{Y|X}(y|\widetilde{x})}$.
Define
\begin{eqnarray}
&&E_{ex}(Q_{X},W_{Y|X},R)\nonumber\\
&&\hspace{0.3in}=\min\limits_{Q_{\widetilde{X}|X}:Q_X=Q_{\widetilde{X}},I(Q_X,Q_{\widetilde{X}|X})\leq
R}\left[\mathbb{E}_{Q_{X\widetilde{X}}}d_{W_{Y|X}}(X,\widetilde{X})+I(Q_X,Q_{\widetilde{X}|X})-R\right],\label{constex}\\
&&E_{rc}(Q_{X},W_{Y|X},R)=\min\limits_{V_{Y|X}}\left[D(V_{Y|X}\|W_{Y|X}|Q_{X})+|I(Q_X,V_{Y|X})-R|^+\right],\label{constrc}\\
&&E_{sp}(Q_X,W_{Y|X},R)=\min\limits_{V_{Y|X}:I(Q_X,V_{Y|X})\leq
R}D(V_{Y|X}\|W_{Y|X}|Q_X),\label{constsp}
\end{eqnarray}
where in (\ref{constex}), $Q_{\widetilde{X}}$ and
$Q_{X\widetilde{X}}$ are respectively the marginal probability
distribution of $\widetilde{X}$ and the joint probability
distribution of $X$ and $\widetilde{X}$ induced by $Q_X$ and
$Q_{\widetilde{X}|X}$.

Let $R^{\infty}_{ex}(Q_X,W_{Y|X})$ be the smallest $R\geq 0$ with
$E_{ex}(Q_{X},W_{Y|X},R)<\infty$. We have
\begin{eqnarray}
R^{\infty}_{ex}(Q_X,W_{Y|X})=\min\limits_{Q_{\widetilde{X}|X}:Q_X=Q_{\widetilde{X}},\mathbb{E}_{Q_{X\widetilde{X}}}d_{W_{Y|X}}(X,\widetilde{X})<\infty}I(Q_X,Q_{\widetilde{X}|X}).\label{Rexinfty}
\end{eqnarray}
It is known \cite[Excercise 5.18]{CK81B} that
$E_{ex}(Q_{X},W_{Y|X},R)$ is a decreasing convex function of $R$
for $R\geq R^{\infty}_{ex}(Q_X,W_{Y|X})$; moreover, the minimum in
(\ref{Rexinfty}) is achieved at $Q_{X\widetilde{X}}$ if and only
if
\begin{eqnarray*}
Q_{X\widetilde{X}}(x,\widetilde{x})=\left\{\begin{array}{ll}
  cQ(x)Q(\widetilde{x})&\mbox{if } d_{W_{Y|X}}(x,\widetilde{x})<\infty,\\
  0
  &\mbox{otherwise},
\end{array}\right.&
\end{eqnarray*}
where the probability distribution $Q$ and the constant $c$ are
uniquely determined by the condition $Q_X=Q_{\widetilde{X}}$.

It is shown in \cite[Lemma 3]{CK81} that, for some
$R^*(Q_{X},W_{Y|X})\in[0, I(Q_{X},W_{Y|X})]$, we have
\begin{eqnarray}
\max\left\{E_{ex}(Q_{X},W_{Y|X},R),E_{rc}(Q_{X},P_{Y|X},R)\right\}=\left\{\begin{array}{ll}
  E_{ex}(Q_{X},W_{Y|X},R)&\mbox{if } R\leq R^*(Q_{X},W_{Y|X}),\\
  E_{rc}(Q_{X},W_{Y|X},R)
  &\mbox{if } R>R^*(Q_{X},W_{Y|X}).
\end{array}\right.&\label{rcex}
\end{eqnarray}
It is also known \cite[Corollary 5.4]{CK81B} that
\begin{eqnarray}
E_{rc}(Q_X,W_{Y|X},R)=\left\{\begin{array}{ll}
  E_{sp}(Q_X,W_{Y|X},R)&\mbox{if } R\geq R_{cr}(Q_{X},W_{Y|X}),\\
  E_{sp}(Q_X,W_{Y|X},R_{cr})+R_{cr}-R
  &\mbox{if } 0\leq R\leq R_{cr}(Q_{X},W_{Y|X}),
\end{array}\right.&\label{rcsp}
\end{eqnarray}
where $R_{cr}\triangleq R_{cr}(Q_{X},W_{Y|X})$ is the smallest $R$
at which the convex curve $E_{sp}(Q_{X},W_{Y|X},R)$ meets its
supporting line of slope -1. It is obvious that
$R_{cr}(Q_{X},W_{Y|X})\leq I(Q_{X},W_{Y|X})$.

\begin{proposition}\label{prop:cr}
$R_{cr}(Q_X,W_{Y|X})=I(Q_X,W_{Y|X})$ if and only if for all $x, y$
such that $Q_X(x)W_{Y|X}(y|x)>0$, the value of
\begin{eqnarray*}
\log\frac{W_{Y|X}(y|x)}{\sum_{x'}Q_X(x')W_{Y|X}(y|x')}
\end{eqnarray*}
does not depend on $y$.
\end{proposition}
\begin{proof}
See Appendix \ref{app:prop:cr}
\end{proof}

Define
$R^{\infty}_{sp}(Q_X,W_{Y|X})=\inf\{R>0:E_{sp}(Q_X,W_{Y|X},R)<\infty\}$.
It is known \cite[Excercise 5.3]{CK81B} that
\begin{eqnarray}
R^{\infty}_{sp}(Q_X,W_{Y|X})=\min I(Q_X,V_{Y|X}),\label{Rspinfty}
\end{eqnarray}
where the minimum is taken over those $V_{Y|X}$'s for which
$V_{Y|X}(y|x)=0$ whenever $W_{Y|X}(y|x)=0$; in particular,
$R^{\infty}_{sp}(Q_X,W_{Y|X})>0$ if and only if for every
$y\in\mathcal{Y}$ there exists an $x\in\mathcal{X}$ with
$Q_X(x)>0$ and $W_{Y|X}(y|x)=0$.

\begin{proposition}\label{prop:sp}
The minimum in (\ref{Rspinfty}) is achieved at $V_{Y|X}=W_{Y|X}$
if and only if the value of
\begin{eqnarray*}
\frac{W_{Y|X}(y|x)}{\sum_{x'} Q_{X}(x')W_{Y|X}(y|x')}
\end{eqnarray*}
does not depend on $y$ for all $x,y$ such that
$Q_X(x)W_{Y|X}(y|x)>0$.
\end{proposition}
\begin{proof}
The proof is similar to that of Proposition \ref{prop:cr}. The
details are omitted.
\end{proof}

One can readily prove the following result by combining
Propositions \ref{prop:cr} and \ref{prop:sp}.
\begin{proposition}
The following statements are equivalent:
\begin{enumerate}
\item $R_{cr}(Q_X,P_{Y|X})=I(Q_X,W_{Y|X})$;

\item $R^{\infty}_{sp}(Q_X,P_{Y|X})=I(Q_X,W_{Y|X})$;

\item for all $x, y$
such that $Q_X(x)W_{Y|X}(y|x)>0$, the value of
\begin{eqnarray*}
\log\frac{W_{Y|X}(y|x)}{\sum_{x'}Q_X(x')W_{Y|X}(y|x')}
\end{eqnarray*}
does not depend on $y$.
\end{enumerate}
\end{proposition}

\begin{proposition}\label{prop:constantbounds}
\begin{enumerate}
\item
$E(Q_X,W_{Y|X},R)\geq\max\{E_{ex}(Q_X,W_{Y|X},R),E_{rc}(Q_X,W_{Y|X},R)\}$;

\item $E(Q_X,P_{Y|X},R)\leq
E_{sp}(Q_X,W_{Y|X},R)$ with the possible exception of
$R=R^{\infty}_{sp}(Q_X,W_{Y|X})$ at which point the inequality not
necessary holds;

\item
$E^c(Q_{X},W_{Y|X},R)=\min_{V_{Y|X}}\left[D(V_{Y|X}\|W_{Y|X}|Q_{X})+|R-I(Q_X;V_{Y|X})|^+\right]$.
\end{enumerate}
\end{proposition}
Remark: $E_{ex}(Q_X,W_{Y|X},R)$, $E_{rc}(Q_X,W_{Y|X},R)$, and
$E_{sp}(Q_X,W_{Y|X},R)$ are respectively the expurgated exponent,
the random coding exponent, and the sphere packing exponent of
channel $W_{Y|X}$ for constant composition codes of type
approximately $Q_X$. The results in Proposition
\ref{prop:constantbounds} are well known. However, bounding the
decoding error probability of constant composition codes often
serves as an intermediate step in characterizing the reliability
function for general block codes; as a consequence, the
reliability function for constant composition codes is rarely
explicitly defined. Moreover, $E_{ex}(Q_X,W_{Y|X},R)$,
$E_{rc}(Q_X,W_{Y|X},R)$, and $E_{sp}(Q_X,W_{Y|X},R)$ are commonly
used to bound the decoding error probability of constant
composition codes for a fixed block length $n$; therefore, it is
implicitly assumed that $Q_X$ is taken from
$\mathcal{P}_n(\mathcal{X})$ (see, e.g., \cite{CK81B}). In
contrast, we consider a sequence of constant composition codes
with block length increasing to infinity and type converging to
$Q_X$ for some $Q_X\in\mathcal{P}(\mathcal{X})$ (see Definition
\ref{def:constantcomposition}). A continuity argument is required
for passing $Q_X$ from $\mathcal{P}_n(\mathcal{X})$ to
$\mathcal{P}(\mathcal{X})$. For completeness, we supply the proof
in Appendix \ref{app:prop:constantbounds}. Note that different
from $E(Q_X,W_{Y|X},\cdot)$, the function $E^c(Q_X,W_{Y|X},\cdot)$
has been completely characterized.

\begin{proposition}\label{prop:constanttogeneral}
\begin{enumerate}
\item $E(W_{Y|X},R)=\sup_{Q_X}E(Q_X,W_{Y|X},R)$,

\item $E^c(W_{Y|X},R)=\inf_{Q_X}E^c(Q_X,W_{Y|X},R)$.
\end{enumerate}
\end{proposition}
Remark: In view of the fact that $E^c(Q_X,W_{Y|X},R)$ is a
continuous function of $Q_X$, we can replace ``$\inf$" with
``$\min$" in the above equation, i.e.,
\begin{eqnarray}
E^c(W_{Y|X},R)=\min\limits_{Q_X}E^c(Q_X,W_{Y|X},R).\label{channelc}
\end{eqnarray}

\begin{proof}
It is obvious that $E(W_{Y|X},R)\geq\sup_{Q_X}E(Q_X,W_{Y|X},R)$;
the other direction follows from the fact that every block code
$\mathcal{C}_n$ contains a constant composition code
$\mathcal{C}'_n$ with $P_{e,\max}(\mathcal{C}'_n,W_{Y|X})\leq
P_{e,\max}(\mathcal{C}_n,W_{Y|X})$ and $R(\mathcal{C}'_n)\geq
R(\mathcal{C}_n)-|\mathcal{X}|\frac{\log(n+1)}{n}$. Similarly, it
is clear that $E^c(W_{Y|X},R)\leq\inf_{Q_X}E^c(Q_X,W_{Y|X},R)$;
the other direction follows from the fact that given any block
code $\mathcal{C}_n$, one can construct a constant composition
code $\mathcal{C}'_n$ with $P_{c}(\mathcal{C}'_n,W_{Y|X})\leq
(n+1)^{|\mathcal{X}|}P_{c}(\mathcal{C}_n,W_{Y|X})$ and
$R(\mathcal{C}'_n)=R(\mathcal{C}_n)$ \cite{DK79}.
\end{proof}

The expurgated exponent, random coding exponent, and sphere
packing exponent of channel $W_{Y|X}$ for general block codes are
defined as follows:
\begin{enumerate}
\item expurgated exponent
\begin{eqnarray}
E_{ex}(W_{Y|X},R)=\max\limits_{Q_X}E_{ex}(Q_X,W_{Y|X},R),\label{channelex}
\end{eqnarray}

\item random coding exponent
\begin{eqnarray}
E_{rc}(W_{Y|X},R)=\max\limits_{Q_X}E_{rc}(Q_X,W_{Y|X},R),\label{channelrc}
\end{eqnarray}

\item sphere packing exponent
\begin{eqnarray}
E_{sp}(W_{Y|X},R)=\max\limits_{Q_X}E_{sp}(Q_X,W_{Y|X},R).
\label{channelsp}
\end{eqnarray}
\end{enumerate}

Let $R^{\infty}_{sp}(W_{Y|X})$ be the smallest $R$ to the right of
which $E_{sp}(W_{Y|X},R)$ is finite. It is known \cite[Excercise
5.3]{CK81B}\cite{Gallager68} that
\begin{eqnarray*}
R^{\infty}_{sp}(W_{Y|X})&=&\max_{Q_X}R^{\infty}_{sp}(Q_X,W_{Y|X})\\
&=&-\log\left[\min\limits_{Q_X}\max\limits_{y}\sum\limits_{x\in\mathcal{X}:W_{Y|X}(y|x)>0}Q_X(x)\right].
\end{eqnarray*}
By Propositions \ref{prop:constantbounds} and
\ref{prop:constanttogeneral}, we recover the following well-known
result \cite{Gallager68, CK81B}:
\begin{eqnarray}
\max\{E_{ex}(W_{Y|X},R),E_{rc}(W_{Y|X},R)\}\leq E(W_{Y|X},R)\leq
E_{sp}(W_{Y|X},R)\label{reliabilitych}
\end{eqnarray}
with the possible exception of $R=R^{\infty}_{sp}(W_{Y|X})$ at
which point the second inequality in (\ref{reliabilitych}) not
necessarily holds.

Now we proceed to review the results on the reliability function
of fixed-rate Slepian-Wolf coding. A fixed-rate Slepian-Wolf code
$\phi_n(\cdot)$ is a mapping from $\mathcal{X}^n$ to a set
$\mathcal{A}_n$. The rate of $\phi_n(\cdot)$ is defined as
\begin{eqnarray*}
R(\phi_n)=\frac{1}{n}\log|\mathcal{A}_n|.
\end{eqnarray*}
Given $\phi_n(X^n)$ and $Y^n$, the output of the optimal
\textit{maximum a posteriori} (MAP) decoder is
\begin{eqnarray*}
\widehat{X}^n&=&\arg\min\limits_{x^n:
\phi_n(x^n)=\phi_n(X^n)}-\sum\limits_{i=1}^n\log
P_{X|Y}(x_i|Y_i)\\
&=&\arg\min\limits_{x^n:
\phi_n(x^n)=\phi_n(X^n)}-\sum\limits_{i=1}^n\log P_{XY}(x_i,Y_i),
\end{eqnarray*}
where the ties are broken in an arbitrary manner. The decoding
error probability of Slepian-Wolf code $\phi_n(\cdot)$ is defined
as
\begin{eqnarray*}
P_e(\phi_n,P_{XY})=\mbox{Pr}\{\widehat{X}^n\neq X^n\}.
\end{eqnarray*}
The correct decoding probability of Slepian-Wolf code
$\phi_n(\cdot)$ is defined as
\begin{eqnarray*}
P_c(\phi_n,P_{XY})=1-P_e(\phi_n,P_{XY}).
\end{eqnarray*}

\begin{definition}\label{def:swfixed}
Given a joint probability distribution $P_{XY}$, we say an error
exponent $E\geq 0$ is achievable with fixed-rate Slepian-Wolf
codes at rate $R$ if for any $\delta>0$, there exists a sequence
of fixed-rate Slepian-Wolf codes $\{\phi_n\}$ such that
\begin{eqnarray*}
&&\limsup\limits_{n\rightarrow\infty}R(\phi_n)\leq R+\delta,\\
&&\limsup\limits_{n\rightarrow\infty}-\frac{1}{n}\log
P_e(\phi_n,P_{XY})\geq E-\delta.
\end{eqnarray*}
The largest achievable error exponent at rate $R$ is denoted by
$E_f(P_{XY},R)$. The function $E_f(P_{XY},\cdot)$ is referred to
as the reliability function of fixed-rate Slepian-Wolf coding.
Similarly, we say a correct decoding exponent $E^c\geq 0$ is
achievable with fixed-rate Slepian-Wolf codes at rate $R$ if for
any $\delta>0$, there exists a sequence of fixed-rate Slepian-Wolf
codes $\{\phi_n\}$ such that
\begin{eqnarray*}
&&\limsup\limits_{n\rightarrow\infty}R(\phi_n)\leq R+\delta,\\
&&\liminf\limits_{n\rightarrow\infty}-\frac{1}{n}\log
P_c(\phi_n,P_{XY})\leq E^c+\delta.
\end{eqnarray*}
The smallest achievable correct decoding exponent at rate $R$ is
denoted by $E^c_f(P_{XY},R)$. It will be seen that
$E^c_f(P_{XY},R)$ is positive if and only if $R<H(X|Y)$.
Therefore, we shall refer to the function $E^c_f(P_{XY},\cdot)$ as
the reliability function of fixed-rate Slepian-Wolf coding below
the Slepian-Wolf limit.
\end{definition}

Fixed-rate Slepian-Wolf coding has been studied extensively
\cite{Gallager76, CK80, CK81, Csiszar82, OH94}. The expurgated
exponent, random coding scheme, and sphere packing exponent of
fixed-rate Slepian-Wolf coding are defined as follows:
\begin{enumerate}
\item expurgated exponent
\begin{eqnarray}
E_{f,ex}(P_{XY},R)=\min\limits_{Q_X}\left[D(Q_X\|P_X)+E_{ex}(Q_X,P_{Y|X},H(Q_X)-R)\right],\label{fixedswex}
\end{eqnarray}

\item random coding exponent
\begin{eqnarray}
E_{f,rc}(P_{XY},R)=\min\limits_{Q_X}\left[D(Q_X\|P_X)+E_{rc}(Q_X,P_{Y|X},H(Q_X)-R)\right],\label{fixedswrc}
\end{eqnarray}

\item sphere packing exponent
\begin{eqnarray}
E_{f,sp}(P_{XY},R)=\min\limits_{Q_X}\left[D(Q_X\|P_X)+E_{sp}(Q_X,P_{Y|X},H(Q_X)-R)\right].\label{fixedswsp}
\end{eqnarray}
\end{enumerate}
Equivalently, the random coding exponent and sphere packing
exponent of fixed-rate Slepian-Wolf coding can be written as
\cite{Gallager76}:
\begin{eqnarray*}
&&E_{rc}(P_{XY}, R)=\max\limits_{0\leq\rho\leq
1}\left\{-\log\sum\limits_{y}\left[\sum\limits_{x}P_{XY}(x,y)^{\frac{1}{1+\rho}}\right]^{1+\rho}+\rho
R\right\},\\
&&E_{sp}(P_{XY},
R)=\sup\limits_{\rho>0}\left\{-\log\sum\limits_{y}\left[\sum\limits_{x}P_{XY}(x,y)^{\frac{1}{1+\rho}}\right]^{1+\rho}+\rho
R\right\}.
\end{eqnarray*}
To see the connection between the random coding exponent and the
sphere packing exponent, we shall write them in the following
parametric forms \cite{Gallager76}:
\begin{eqnarray*}
&&R=H(X^{(\rho)}|Y^{(\rho)}),\\
&&E_{f,sp}(P_{XY},R)=D(P_{X^{(\rho)}Y^{(\rho)}}\|P_{XY}),
\end{eqnarray*}
and
\begin{eqnarray*}
E_{rc}(P_{XY}, R)=\left\{\begin{array}{ll}
  D(P_{X^{(\rho)}Y^{(\rho})}\|P_{XY})&\mbox{if } \left. H(X|Y)\leq R\leq H(X_{\rho}|Y_{\rho})\right|_{\rho=1},\\
  -\log\sum\limits_{y}\left[\sum\limits_{x}\sqrt{P_{XY}(x,y)}\right]^2+R  &\mbox{if } R>\left.
  H(X^{(\rho)}|Y^{(\rho)})\right|_{\rho=1},
\end{array}\right.&\\
\end{eqnarray*}
where the joint distribution of $(X^{(\rho)},Y^{(\rho)})$ is
$P_{X^{(\rho)}Y^{(\rho)}}$, which is specified by
\begin{eqnarray}
&&P_{Y^{(\rho)}}(y)=\frac{P_Y(y)\left[\sum_{x}P_{X|Y}(x|y)^{\frac{1}{1+\rho}}\right]^{1+\rho}}{\sum_{y'}P_{Y}(y')\left[\sum_{x}P_{X|Y}(x|y')^{\frac{1}{1+\rho}}\right]^{1+\rho}},\quad y\in\mathcal{Y},\label{para1}\\
&&P_{X^{(\rho)}|Y^{(\rho)}}(x|y)=\frac{P_{X|Y}(x|y)^{\frac{1}{1+\rho}}}{\sum_{x'}P_{X|Y}(x'|y)^{\frac{1}{1+\rho}}},\quad
x\in\mathcal{X},y\in\mathcal{Y}.\label{para2}
\end{eqnarray}
Define the critical rate
\begin{eqnarray*}
R_{f,cr}(P_{XY})=\left.H(X^{(\rho)}|Y^{(\rho)})\right|_{\rho=1}.
\end{eqnarray*}
Note that $E_{rc}(P_{XY},R)$ and $E_{sp}(P_{XY},R)$ coincide when
$R\in[H(X|Y), R_{f,cr}(P_{XY})]$. Let
$R^{\infty}_{f,sp}(P_{XY})=\sup\{R:E_{f,sp}(P_{XY},R)<\infty\}$.
It is shown in \cite{CHJL07ISIT} that
\begin{eqnarray*}
R^{\infty}_{f,sp}(P_{XY})=\max_{y}\log|\{x\in\mathcal{X}:P_{X|Y}(x|y)>0\}|.
\end{eqnarray*}
It is well known that the reliability function $E_f(P_{XY},\cdot)$
is upper-bounded by $E_{f,sp}(P_{XY},\cdot)$ and lower-bounded by
$E_{f,rc}(P_{XY},\cdot)$ and $E_{f,ex}(P_{XY},\cdot)$
\cite{Gallager76, CK80, CK81}, i.e.,
\begin{eqnarray}
\max\{E_{f,rc}(P_{XY},R),E_{f,ex}(P_{XY},R)\}\leq
E_{f}(P_{XY},R)\leq E_{f,sp}(P_{XY},R)\label{fixedsw}
\end{eqnarray}
with the possible exception of $R=R^{\infty}_{f,sp}(P_{XY})$ at
which point the second inequality in (\ref{fixedsw}) not
necessarily holds. Note that $E_{f}(P_{XY},R)$ is completely
characterized for $R\in[H(X|Y), R_{f,cr}(P_{XY})]$.

Unlike $E_f(P_{XY},\cdot)$, the function $E^c_f(P_{XY},\cdot)$ has
been characterized for all $R$. Specifically, it is shown in
\cite{OH94,CHJL07} that
\begin{eqnarray}
E^c_{f}(P_{XY},R)=\min\limits_{Q_{X}}\left[D(Q_{X}\|P_X)+E^c(Q_{X},P_{Y|X},H(Q_X)-R)\right].\label{fixedswc}
\end{eqnarray}

Comparing  (\ref{channelex}) with (\ref{fixedswex}),
(\ref{channelrc}) with (\ref{fixedswrc}), (\ref{channelsp}) with
(\ref{fixedswsp}), and (\ref{channelc})  with (\ref{fixedswc}),
one can easily see that there exists an intimate connection
between fixed-rate Slepian-Wolf coding for source distribution
$P_{XY}$ and channel coding for channel $P_{Y|X}$. This connection
can be roughly interpreted as the manifestation of the following
facts \cite{AD82}.
\begin{enumerate}
\item Given, for each type $Q_X\in\mathcal{P}_n(\mathcal{X})$, a constant composition code
$\mathcal{C}_n(Q_X)\subseteq\mathcal{T}_n(Q_X)$ with
$R(\mathcal{C}_n(Q_X))\approx H(Q_X)-R$ and
$P_{e,\max}(\mathcal{C}_n(Q_X),P_{Y|X})\approx e^{-nE(Q_X)}$, one
can use $\mathcal{C}_n(Q_X)$ to partition type class
$\mathcal{T}_n(Q_X)$ into approximately $e^{nR}$ disjoint subsets
such that each subset is a constant composition code of type $Q_X$
with the maximum decoding error probability over channel $P_{Y|X}$
approximately equal to or less than that of $\mathcal{C}_n(Q_X)$.
Note that these partitions, one for each type class, yield a
fixed-rate Slepian-Wolf code of rate approximately $R$ with
$\mbox{Pr}\{\widehat{X}^n\neq
X^n|X^n\in\mathcal{T}_n(Q_X)\}\lessapprox e^{-nE(Q_X)}$. Since
$\mbox{Pr}\{X^n\in\mathcal{T}_n(Q_X)\}\approx e^{-nD(Q_X\|P_X)}$
(cf. (\ref{typeclass}), (\ref{typeprobability})), it follows that
$\mbox{Pr}\{\widehat{X}^n\neq
X^n,X^n\in\mathcal{T}_n(Q_X)\}\lessapprox
e^{-n[D(Q_X\|P_X)+E(Q_X)]}$. The overall decoding error
probability $\mbox{Pr}\{\widehat{X}^n\neq X^n\}$ of the resulting
Slepian-Wolf code can be upper-bounded, on the exponential scale,
by $e^{-n[D(Q^*_X\|P_X)+E(Q^*_X)]}$, where
$Q^*_X=\arg\min_{Q_X}D(Q_X\|P_X)+E(Q_X)$. In contrast, one has the
freedom to choose $Q_X$ in channel coding, which explains why
maximization (instead of minimization) is used in
(\ref{channelex}), (\ref{channelrc}), and (\ref{channelsp}).

\item Given a fixed-rate Slepian-Wolf code $\phi_n(\cdot)$ with $R(\phi_n)\approx R$
and $P_{e}(\phi_n,P_{XY})\approx e^{-nE}$, one can, for each type
$Q_X\in\mathcal{P}_n(\mathcal{X})$, lift out a constant
composition code $\mathcal{C}_n(Q_X)\subseteq\mathcal{T}_n(Q_X)$
with $R(\mathcal{C}_n(Q_X))\gtrapprox H(Q_X)-R$ and
$P_{e}(\mathcal{C}_n(Q_X),P_{Y|X})\lessapprox
e^{-n[E-D(Q_X\|P_X)]}$.

\item The correct decoding exponents for channel coding and fixed-rate Slepian-Wolf coding can be interpreted in a similar way.
Note that in channel coding, to maximize the correct decoding
probability one has to minimize the correct decoding exponent;
this is why in (\ref{channelc}) minimization (instead of
maximization) is used.
\end{enumerate}

Therefore, it should be clear that to characterize the reliability
functions for channel coding and fixed-rate Slepian-Wolf coding,
it suffices to focus on constant composition codes. It will be
shown in the next section that a similar reduction holds for
variable-rate Slepian-Wolf coding. Indeed, the reliability
function for constant component codes plays a predominant role in
determining the fundamental rate-error tradeoff in variable-rate
Slepian-Wolf coding.

\section{Variable-Rate Slepian-Wolf Coding: Above the Slepian-Wolf Limit}\label{variablerate}

A variable-rate Slepian-Wolf code $\varphi_n(\cdot)$ is a mapping
from $\mathcal{X}^n$ to a binary prefix code $\mathcal{B}_n$. Let
$l(\phi_n(x^n))$ denote the length of binary string $\phi_n(x^n)$.
The rate\footnote{It is worth noting that $R(\varphi_n,P_{XY})$
depends on $P_{XY}$ only through $P_X$.} of variable-rate
Slepian-Wolf code $\phi_n(\cdot)$ is defined as
\begin{eqnarray*}
R(\varphi_n,P_{XY})=\frac{1}{n\log_2e}\mathbb{E}l(\varphi_n(X^n)).
\end{eqnarray*}
Given $\varphi_n(X^n)$ and $Y^n$, the output of the optimal
\textit{maximum a posteriori} (MAP) decoder is
\begin{eqnarray*}
\widehat{X}^n&=&\arg\min\limits_{x^n:
\varphi_n(X^n)=\varphi_n(X^n)}-\sum\limits_{i=1}^n\log
P_{X|Y}(x_i|Y_i)\\
&=&\arg\min\limits_{x^n:
\varphi_n(X^n)=\varphi_n(X^n)}-\sum\limits_{i=1}^n\log
P_{XY}(x_i,Y_i),
\end{eqnarray*}
where the ties are broken in an arbitrary manner. The decoding
error probability of variable-rate Slepian-Wolf code
$\varphi_n(\cdot)$ is defined as
\begin{eqnarray*}
P_e(\varphi_n,P_{XY})=\mbox{Pr}\{\widehat{X}^n\neq X^n\}.
\end{eqnarray*}
The correct decoding probability of Slepian-Wolf code
$\varphi_n(\cdot)$ is defined as
\begin{eqnarray*}
P_c(\phi_n,P_{XY})=1-P_e(\varphi_n,P_{XY}).
\end{eqnarray*}

\begin{definition}\label{def:swvariable}
Given a joint probability distribution $P_{XY}$, we say an error
exponent $E\geq 0$ is achievable with variable-rate Slepian-Wolf
codes at rate $R$ if for any $\delta>0$, there exists a sequence
of variable-rate Slepian-Wolf codes $\{\varphi_n\}$ such that
\begin{eqnarray*}
&&\limsup\limits_{n\rightarrow\infty}R(\varphi_n,P_{XY})\leq R+\delta,\\
&&\limsup\limits_{n\rightarrow\infty}-\frac{1}{n}\log
P_e(\varphi_n,P_{XY})\geq E-\delta.
\end{eqnarray*}
The largest achievable error exponent at rate $R$ is denoted by
$E_v(P_{XY},R)$. The function $E_v(P_{XY},\cdot)$ is referred to
as the reliability function of variable-rate Slepian-Wolf coding.
\end{definition}

The power of variable-rate Slepian-Wolf coding results from its
flexibility in rate allocation. Note that in fixed-rate
Slepian-Wolf coding, one has to allocate the same amount of rate
to each type class\footnote{Since there are only polynomial number
of types for any given $n$ (cf. (\ref{cardinality})), the encoder
can convey the type information to the decoder using negligible
amount of rate when $n$ is large enough. Therefore, without loss
of much generality, we can assume that the type of $X^n$ is known
to the decoder. Under this assumption, an optimal fixed-rate
Slepian-Wolf encoder of rate $R$ should partition
$\mathcal{T}_n(P)$ into $\min\{|\mathcal{T}_n(P)|,e^{nR}\}$
disjoint subsets for each $P\in\mathcal{P}_n$. It can be seen that
the rate allocated to $\mathcal{T}_n(P)$ is always $R$ if
$|\mathcal{T}_n(P)|\geq e^{nR}$.}. In general, the type $Q^*_X$
that dominates the error probability of fixed-rate Slepian-Wolf
coding is different from $P_X$. In contrast, for variable-rate
Slepian-Wolf coding, we can losslessly compress the sequences of
types that are bounded away $P_X$ by allocating enough rate to
those type classes (but its contribution to the overall rate is
still negligible since the probability of those type classes are
extremely small), and therefore, effectively eliminate the
dominant error event in fixed-rate Slepian-Wolf coding. As a
consequence, the types that can cause decoding error in
variable-rate Slepian-Wolf coding must be very close to $P_X$.
This is the main intuition underlying the proof of the following
theorem.

\begin{theorem}\label{theorem1}
$E_v(P_{XY},R)=E(P_X,P_{Y|X},H(P_X)-R)$.
\end{theorem}
\begin{proof}
The proof is divided into two parts. Firstly, we shall show that
$E_v(P_{XY},R)\geq E(P_X,P_{Y|X},H(P_X)-R)$. The main idea is that
one can use a constant composition code $\mathcal{C}_n$ of type
approximately $P_X$ and rate approximately $H(P_X)-R$ to construct
a variable-rate Slepian-Wolf code $\varphi_{n'}(\cdot)$ with
$n'\approx n$, $R(\varphi_{n'},P_{XY})\approx R$, and
$P_e(\varphi_{n'},P_{XY})\leq P_{e,\max}(\mathcal{C}_n,P_{Y|X})$.

By Definition \ref{def:constantcomposition}, for any $\delta>0$,
there exists a sequence of constant composition codes
$\{\mathcal{C}_n\}$ with
$\mathcal{C}_n\subseteq\mathcal{T}_n(P_n)$ for some
$P_n\in\mathcal{P}_n(\mathcal{X})$ such that
\begin{eqnarray*}
&&\lim\limits_{n\rightarrow\infty}\|P_n-P_X\|=0,\\
&&\liminf_{n\rightarrow\infty}R(\mathcal{C}_n)\geq
H(P_X)-R-\delta,\\
&&\limsup_{n\rightarrow\infty}\frac{1}{n}\log
P_{e,\max}(\mathcal{C}_n,P_{Y|X})\geq
E(P_X,P_{Y|X},H(P_X)-R)-\delta.
\end{eqnarray*}
Since $P_X(x)>0$ for all $x\in\mathcal{X}$, we have
\begin{eqnarray*}
\max\limits_{P\in\mathcal{P}_n(\mathcal{X})\cap\mathcal{E}(\delta)}\max\limits_{x}\frac{P_n(x)}{P(x)}\leq
(1+\delta)^2
\end{eqnarray*}
for all sufficiently $n$, where
\begin{eqnarray*}
\mathcal{E}(\delta)=\left\{P\in\mathcal{P}(\mathcal{X}):\max\limits_{x}\frac{P_X(x)}{P(x)}\leq
1+\delta, H(P)\leq H(P_X)+\delta, D(P\|P_X)\leq\delta\right\}.
\end{eqnarray*}
Let $k_n=\lceil(1+\delta)^2n\rceil$. When $n$ is large enough, we
can, for each
$P\in\mathcal{P}_{k_n}(\mathcal{X})\cap\mathcal{E}(\delta)$,
construct a constant composition code $\mathcal{C}'_{k_n}(P)$ of
length $k_n$ and type $P$ by concatenating a fixed sequence in
$\mathcal{X}^{k_n-n}$ to each codeword in $\mathcal{C}_{n}$. It is
easy to see that
\begin{eqnarray}
&&|\mathcal{C}'_{k_n}(P)|=|\mathcal{C}_n|,\label{cardinv}\\
&&P_{e,\max}(\mathcal{C}'_{k_n}(P),P_{Y|X})=P_{e,\max}(\mathcal{C}_n,P_{Y|X})
\label{appendinv}
\end{eqnarray}
for all
$P\in\mathcal{P}_{k_n}(\mathcal{X})\cap\mathcal{E}(\delta)$. One
can readily show by invoking the covering lemma in
\cite{Ahlswede80} that for each
$P\in\mathcal{P}_{k_n}(\mathcal{X})\cap\mathcal{E}(\delta)$, there
exist $L(k_n)$ permutations $\pi_1,\cdots,\pi_{L(k_n)}$ of the
integers $1,\cdots,k_n$ such that
\begin{eqnarray*}
\bigcup_{i=1}^{L(k_n)}\pi_i(\mathcal{C}'_{k_n}(P))=\mathcal{T}_{k_n}(P),
\end{eqnarray*}
where
\begin{eqnarray*}
L(k_n)&\triangleq&\max\limits_{P\in\mathcal{P}_{k_n}(\mathcal{X})\cap\mathcal{E}(\delta)}\left\lfloor|\mathcal{C}'_{k_n}(P)|^{-1}|\mathcal{T}_{k_n}(P)|\log|\mathcal{T}_{k_n}(P)|+1\right\rfloor.
\end{eqnarray*}
In view of (\ref{cardinv}), we can rewrite $L(k_n)$ as
\begin{eqnarray*}
L(k_n)=\max\limits_{P\in\mathcal{P}_{k_n}(\mathcal{X})\cap\mathcal{E}(\delta)}\left\lfloor|\mathcal{C}_n|^{-1}|\mathcal{T}_{k_n}(P)|\log|\mathcal{T}_{k_n}(P)|+1\right\rfloor.
\end{eqnarray*}
Note that
\begin{eqnarray}
P_{e,\max}(\pi_i(\mathcal{C}'_{k_n}(P)),P_{Y|X})=P_{e,\max}(\mathcal{C}'_{k_n}(P),P_{Y|X}),\quad
i=1,2,\cdots,L(k_n). \label{permutationinv}
\end{eqnarray}
Given
$\pi_1(\mathcal{C}'_{k_n}(P)),\cdots,\pi_{L(k_n)}(\mathcal{C}'_{k_n}(P))$,
we can partition $\mathcal{T}_{k_n}(P)$ into $L(k_n)$ disjoint
subsets:
\begin{eqnarray*}
&&\mathcal{T}_{k_n}(P,1)=\pi_1(\mathcal{C}'_{k_n}(P)),\\
&&\mathcal{T}_{k_n}(P,i)=\pi_i(\mathcal{C}'_{k_n}(P))\left\backslash\bigcup_{j=1}^{i-1}\pi_i(\mathcal{C}'_{k_n}(P)),\right.\quad
i=2,\cdots,L(k_n).
\end{eqnarray*}
It is clear that
\begin{eqnarray}
P_{e,\max}(\mathcal{T}_{k_n}(P,i),P_{Y|X})\leq
P_{e,\max}(\pi_i(\mathcal{C}'_{k_n}(P)),P_{Y|X}),\quad
i=1,2,\cdots,L(k_n).\label{partitioninv}
\end{eqnarray}

Now construct a sequence of variable-rate Slepian-Wolf codes
$\{\phi_{k_n}(\cdot)\}$ as follows.
\begin{enumerate}
\item The encoder sends the type of $x^{k_n}$ to the decoder, where each type is uniquely represented by a binary sequence of length
$m_1(k_n)$.

\item If $x^{k_n}\in\mathcal{T}_{k_n}(P)$ for some
$P\notin\mathcal{E}(\delta)$, the encoder sends $x^{k_n}$
losslessly to the decoder, where each $x^n\in\mathcal{T}_{k_n}(P)$
is uniquely represented by a binary sequence of length $m_2(k_n)$.

\item If $x^{k_n}\in\mathcal{T}_{k_n}(P)$ for some
$P\in\mathcal{E}(\delta)$, the encoder finds the set
$\pi_{i^*}(\mathcal{C}'_{k_n}(P))$ that contains $x^{k_n}$ and
sends the index $i^*$ to the decoder, where each index in
$\{1,2,\cdots,L(k_n)\}$ is uniquely represented by a binary
sequence of length $m_3(k_n)$.
\end{enumerate}
Specifically, we choose
\begin{eqnarray*}
&&m_1(k_n)=\lceil\log_2|\mathcal{P}_{k_n}(\mathcal{X})|\rceil,\\
&&m_2(k_n)=\max\limits_{P\in\mathcal{P}_{k_n}(\mathcal{X})}\lceil\log_2|\mathcal{T}_{k_n}(P)|\rceil,\\
&&m_3(k_n)=\lceil\log_2L(k_n)\rceil.
\end{eqnarray*}
Note that
\begin{eqnarray*}
R(\varphi_{k_n},P_{XY})&=&\frac{m_1(k_n)+\theta
m_2(k_n)+(1-\theta)m_3(k_n)}{k_n\log_2e},
\end{eqnarray*}
where
\begin{eqnarray*}
\theta=\sum\limits_{P\in\mathcal{P}_{k_n}(\mathcal{X})\cap\mathcal{E}(\delta)}\mbox{Pr}\{X^{k_n}\in\mathcal{T}_{k_n}(P)\}.
\end{eqnarray*}
It is easy to verify (cf. (\ref{cardinality}), (\ref{typeclass})
and (\ref{typeprobability})) that
\begin{eqnarray*}
&&m_1(k_n)\leq|\mathcal{X}|\log_2(k_n+1)+1,\\
&&m_2(k_n)\leq k_n\log_2|\mathcal{X}|+1,\\
&&\theta\leq(k_n+1)^{|\mathcal{X}|}e^{-k_n\delta}.
\end{eqnarray*}
Therefore, we have
\begin{eqnarray}
\limsup\limits_{n\rightarrow\infty}R(\phi_{k_n},P_{XY})&=&\limsup\limits_{n\rightarrow\infty}\frac{m_3(k_n)}{k_n\log_2e}\nonumber\\
&\leq&\max\limits_{P\in\mathcal{E}(\delta)}H(P)-\frac{1}{(1+\delta)^2}\liminf_{n\rightarrow\infty}R(\mathcal{C}_n)\nonumber\\
&\leq&
H(P_X)+\delta-\frac{H(P_X)-R-\delta}{(1+\delta)^2}.\label{rate1}
\end{eqnarray}
By (\ref{appendinv}), (\ref{permutationinv}), (\ref{partitioninv})
and the construction of $\varphi_{k_n}(\cdot)$, it is clear that
\begin{eqnarray*}
P_e(\varphi_{k_n},P_{XY})&=&\sum\limits_{P\in\mathcal{P}_{k_n}(\mathcal{X})\cap\mathcal{E}(\delta)}\sum\limits_{i=1}^{L(k_n)}\mbox{Pr}\{X^{k_n}\in\mathcal{T}_{k_n}(P,i)\}\mbox{Pr}\{\widehat{X}^n\neq X^n|X^{k_n}\in\mathcal{T}_{k_n}(P,i)\}\\
&\leq&\sum\limits_{P\in\mathcal{P}_{k_n}(\mathcal{X})\cap\mathcal{E}(\delta)}\sum\limits_{i=1}^{L(k_n)}\mbox{Pr}\{X^{k_n}\in\mathcal{T}_{k_n}(P,i)\}P_{e,\max}(\mathcal{T}_{k_n}(P,i),P_{Y|X})\\
&\leq&\sum\limits_{P\in\mathcal{P}_{k_n}(\mathcal{X})\cap\mathcal{E}(\delta)}\sum\limits_{i=1}^{L(k_n)}\mbox{Pr}\{X^{k_n}\in\mathcal{T}_{k_n}(P,i)\}P_{e,\max}(\pi_i(\mathcal{C}'_{k_n}(P)),P_{Y|X})\\
&=&\sum\limits_{P\in\mathcal{P}_{k_n}(\mathcal{X})\cap\mathcal{E}(\delta)}\sum\limits_{i=1}^{L(k_n)}\mbox{Pr}\{X^{k_n}\in\mathcal{T}_{k_n}(P,i)\}P_{e,\max}(\mathcal{C}_n,P_{Y|X})\\
&\leq& P_{e,\max}(\mathcal{C}_n,P_{Y|X}),
\end{eqnarray*}
which implies
\begin{eqnarray}
\limsup\limits_{n\rightarrow\infty}-\frac{1}{k_n}\log
P_e(\varphi_{k_n},P_{XY})&\geq&\limsup\limits_{n\rightarrow\infty}-\frac{1}{k_n}\log
P_{e,\max}(\mathcal{C}_n,P_{Y|X})\nonumber\\
&\geq&\frac{E(P_X,P_{Y|X},H(P_X)-R)-\delta}{(1+\delta)^2}\label{exponent}.
\end{eqnarray}
In view of (\ref{rate1}), (\ref{exponent}) and the fact that
$\delta>0$ is arbitrary, we must have $E_v(P_{XY},R)\geq
E(P_X,P_{Y|X},H(P_X)-R)$ (cf. Definition \ref{def:swvariable}).

Now we proceed to show that $E_v(P_{XY},R)\leq
E(P_X,P_{Y|X},H(P_X)-R)$. The main idea is that one can extract a
constant composition code of type approximately $P_X$ and rate
approximately $H(X)-R$ or greater from a given variable-rate
Slepian-Wolf code $\varphi_n(\cdot)$ of rate approximately $R$
such that the average decoding error probability of this constant
composition code over channel $P_{Y|X}$ is bounded from above by
$\gamma^* P_e(\varphi_n,P_{XY})$, where $\gamma^*$ is a constant
that does not depend on $n$.

By Definition \ref{def:swvariable}, for any $\delta>0$, there
exists a sequence of variable-rate Slepian-Wolf codes
$\{\varphi_n\}$ such that
\begin{eqnarray}
&&\limsup\limits_{n\rightarrow\infty}R(\varphi_n,P_{XY})\leq
R+\delta,\label{rate}\\
&&\limsup_{n\rightarrow\infty}-\frac{1}{n}\log
P_e(\varphi_n,P_{XY})\geq E_v(P_{XY},R)-\delta.\label{error}
\end{eqnarray}
Suppose $\varphi_n(\cdot)$ induces a partition\footnote{The
partition is defined as follows:
$\varphi_n(x^n)=\varphi_n(\widetilde{x}^n)$ if
$x^n,\widetilde{x}^n\in\mathcal{T}_n(P,i)$ for some $i$, and
$\varphi_n(x^n)\neq\varphi_n(\widetilde{x}^n)$ if
$x^n\in\mathcal{T}_n(P,i),\widetilde{x}^n\in\mathcal{T}_n(P,j)$
for $i\neq j$.} of $\mathcal{T}_n(P)$,
$P\in\mathcal{P}_n(\mathcal{X})$, into $N_n(P)$ disjoint subsets
$\mathcal{T}_n(P,1),\cdots,\mathcal{T}_n(P,N_n(P))$. Define
\begin{eqnarray*}
&&\mathcal{F}_n(\delta)=\left\{(P,i):
\frac{1}{n}\log\frac{|\mathcal{T}_n(P)|}{|\mathcal{T}_n(P,i)|}\leq
R+2\delta, P\in\mathcal{P}_n(\mathcal{X}),i=1,2,\cdots,N_n(P)\right\},\\
&&\mathcal{G}_{n}(\gamma)=\left\{(P,i):\mbox{Pr}\{\widehat{X}^n\neq
X^n |X^n\in\mathcal{T}_n(P,i)\}\leq \gamma P_e(\varphi_n,P_{XY}),
P\in\mathcal{P}_n(\mathcal{X}),i=1,2,\cdots,N_n(P)\right\},
\end{eqnarray*}
where $\gamma>0$. One can readily verify that
\begin{eqnarray}
&&\sum\limits_{(P,i)\in\mathcal{F}_n(\delta)}\mbox{Pr}\{X^n\in\mathcal{T}_n(P,i)\}\geq
1-\frac{R(\varphi_n,P_{XY})}{R+2\delta},\label{Fn}\\
&&\sum\limits_{(P,i)\in\mathcal{G}_{n}(\gamma)}\mbox{Pr}\{X^n\in\mathcal{T}_n(P,i)\}\geq\frac{\gamma-1}{\gamma}.\label{Gn}
\end{eqnarray}
Moreover, by (\ref{rate}) and (\ref{Fn}) we have
\begin{eqnarray}
\liminf\limits_{n\rightarrow\infty}\sum\limits_{(P,i)\in\mathcal{F}_n(\delta)}\mbox{Pr}\{X^n\in\mathcal{T}_n(P,i)\}\geq\frac{\delta}{R+2\delta}.\label{An}
\end{eqnarray}

Let $\gamma^*$ be a positive number satisfying
\begin{eqnarray*}
\frac{\gamma^*-1}{\gamma^*}+\frac{\delta}{R+2\delta}>1.
\end{eqnarray*}
Define
\begin{eqnarray*}
&&\mathcal{S}_n(\delta)=\left\{P\in\mathcal{P}_n(\mathcal{X}):
H(P)\geq H(P_X)-\delta, \max\limits_{x}\frac{P(x)}{P_X(x)}\leq
1+\delta\right\},\\
&&\mathcal{D}_n(\delta,\gamma^*)=\{(P,i):(P,i)\in\mathcal{F}_n(\delta)\cap\mathcal{G}_{n}(\gamma^*),P\in\mathcal{S}_n(\delta)\}.
\end{eqnarray*}
It follows from the weak law of large numbers that
\begin{eqnarray}
\lim\limits_{n\rightarrow\infty}\sum_{P\in\mathcal{S}_n(\delta)}\mbox{Pr}\{X^n\in\mathcal{T}_n(P)\}=1.\label{Sn}
\end{eqnarray}
In view of (\ref{Gn}), (\ref{An}) and (\ref{Sn}), we have
\begin{eqnarray*}
&&\liminf\limits_{n\rightarrow\infty}\sum\limits_{(P,i)\in\mathcal{D}_n(\delta,\gamma^*)}\mbox{Pr}\{X^n\in\mathcal{T}_n(P,i)\}\\
&&\geq\liminf\limits_{n\rightarrow\infty}\left\{
1-\left[1-\sum\limits_{(P,i)\in\mathcal{F}_n(\delta)}\mbox{Pr}\{X^n\in\mathcal{T}_n(P,i)\}\right]-\left[1-\sum\limits_{(P,i)\in\mathcal{G}_{n}(\gamma^*)}\mbox{Pr}\{X^n\in\mathcal{T}_n(P,i)\}\right]\right.\\
&&\hspace{0.676in}\left.-\left[1-\sum_{P\in\mathcal{S}_n(\delta)}\mbox{Pr}\{X^n\in\mathcal{T}_n(P)\}\right]\right\}\\
&&\geq\frac{\gamma^*-1}{\gamma^*}+\frac{\delta}{R+2\delta}-1\\
&&>0.
\end{eqnarray*}
Therefore, $\mathcal{D}_n(\delta,\gamma^*)$ is non-empty for all
sufficiently large $n$. Pick an arbitrary $(P^*_n,i^*)$ from
$\mathcal{D}_n(\delta,\gamma^*)$ for each sufficiently large $n$.
We can construct a constant composition code $\mathcal{C}_{m_n}$
of length $m_n=\lceil(1+\delta)n\rceil$ and type $P_{m_n}$ for
some $P_{m_n}\in\mathcal{P}_{m_n}(\mathcal{X})$ by concatenating a
fixed sequence in $\mathcal{X}^{m_n-n}$ to each sequence in
$\mathcal{T}_{n}(P^*_n,i^*)$ such that
\begin{eqnarray}
\lim\limits_{n\rightarrow\infty}\|P_{m_n}-P_X\|=0.\label{type2}
\end{eqnarray}
Note that
\begin{eqnarray}
\liminf\limits_{n\rightarrow\infty}R(\mathcal{C}_{m_n})&=&\liminf\limits_{n\rightarrow\infty}\frac{1}{m_n}\log|\mathcal{T}_n(P^*_n,i^*)|\nonumber\\
&\geq&\liminf\limits_{n\rightarrow\infty}\frac{n}{m_n}\left[\frac{1}{n}\log|\mathcal{T}_n(P^*_n)|-R-2\delta\right]\nonumber\\
&\geq&\frac{H(P_X)-R-3\delta}{1+\delta}.\label{rate2}
\end{eqnarray}
Moreover, since
\begin{eqnarray*}
P_e(\mathcal{C}_{m_n},P_{Y|X})=\mbox{Pr}\{\widehat{X}^n\neq
X^n|X^n\in\mathcal{T}_n(P^*_n,i^*)\}\leq\gamma^*P_e(\varphi_n,P_{XY}),
\end{eqnarray*}
it follows from (\ref{error}) that
\begin{eqnarray}
\limsup\limits_{n\rightarrow\infty}-\frac{1}{m_n}\log
P_e(\mathcal{C}_{m_n},P_{Y|X})\geq\frac{E_v(P_{XY},R)-\delta}{1+\delta}.\label{exponent2}
\end{eqnarray}
In view of (\ref{type2}), (\ref{rate2}), (\ref{exponent2}), and
the fact that $\delta>0$ is arbitrary, we must have
$E_v(P_{XY},R)\leq E(P_X,P_{Y|X},H(P_X)-R)$ (cf. Definition
\ref{def:constantcomposition}). The proof is complete.
\end{proof}

The following result is an immediate consequence of Theorem
\ref{theorem1} and Proposition \ref{prop:constantbounds}.
\begin{corollary}
Define
\begin{eqnarray*}
&&E_{v,ex}(P_{XY},R)=E_{ex}(P_X,P_{Y|X},H(P_X)-R),\\
&&E_{v,rc}(P_{XY},R)=E_{rc}(P_X,P_{Y|X},H(P_X)-R),\\
&&E_{v,sp}(P_{XY},R)=E_{sp}(P_X,P_{Y|X},H(P_X)-R).
\end{eqnarray*}
We have
\begin{enumerate}
\item
$E_v(P_{XY},R)\geq\max\{E_{v,ex}(P_{XY},R),E_{v,rc}(P_{XY},R)\}$;

\item
$E_v(P_{XY},R)\leq E_{v,sp}(P_{XY},R)$ with the possible exception
of $R=H(P_X)-R^{\infty}_{sp}(P_X,P_{Y|X})$ at which point the
inequality not necessarily holds.
\end{enumerate}
\end{corollary}

Remark:
\begin{enumerate}
\item
We have $E_{v}(P_{XY},R)=\infty$ for
$R>H(P_X)-R^{\infty}_{ex}(P_X,P_{Y|X})$, and
$E_{v}(P_{XY},R)<\infty$ for
$R<H(P_X)-R^{\infty}_{sp}(P_X,P_{Y|X})$. Therefore,
$H(P_X)-R^{\infty}_{ex}(P_X,P_{Y|X})$ and
$H(P_X)-R^{\infty}_{sp}(P_X,P_{Y|X})$ are respectively the upper
bound and the lower bound on the zero-error rate of variable-rate
Slepian-Wolf coding.

\item In view of (\ref{rcsp}), we have
\begin{eqnarray*}
E_v(P_{XY},R)=E_{v,sp}(P_{XY},R)=E_{sp}(P_X,P_{Y|X},H(P_X)-R)
\end{eqnarray*}
for $R\in[H(X|Y),H(P_X)-R_{cr}(P_{X},P_{Y|X})]$. Note that
\begin{eqnarray*}
E_{v,sp}(P_{XY},R)\geq E_{f,sp}(P_{XY},R)\geq E_{f}(P_{XY},R),
\end{eqnarray*}
where the first inequality is strict unless the minimum in
(\ref{fixedswsp}) is achieved at $Q_X=P_X$, (i.e.,
$P_{X^{(\rho)}}=P_X$, where $P_{X^{(\rho)}}$ is the marginal
distribution of $X^{(\rho)}$ induced by $P_{Y^{(\rho)}}$ and
$P_{X^{(\rho)}|Y^{(\rho)}}$ in (\ref{para1}), (\ref{para2})).
Therefore, variable-rate Slepian-Wolf coding can outperform
fixed-rate Slepian-Wolf coding in terms of rate-error tradeoff.
\end{enumerate}

For $R>H(P_X)-R_{cr}(P_{X},P_{Y|X})$, it is possible to obtain
upper bounds on $E_v(P_{XY},R)$ that are tighter than
$E_{v,sp}(P_{XY},R)$. Let $E_{ex}(P_{Y|X},R)$ and
$E_{sp}(P_{Y|X},R)$ be respectively the expurgated exponent and
the sphere packing exponent of channel $P_{Y|X}$. The
straight-line exponent $E_{sl}(P_{Y|X},R)$ of channel $P_{Y|X}$
\cite{Gallager68} is the smallest linear function of $R$ which
touches the curve $E_{sp}(P_{Y|X},R)$ and also satisfies
\begin{eqnarray*}
E_{sl}(P_{Y|X},0)=E_{ex}(P_{Y|X},0),
\end{eqnarray*}
where $E_{ex}(P_{Y|X},0)$ is assumed to be finite. Let
$R_{sl}(P_{Y|X})$ be the point at which $E_{sl}(P_{Y|X},R)$ and
$E_{sp}(P_{Y|X},R)$ coincide. It is well known \cite{Gallager68}
that $E(P_{Y|X},R)\leq E_{sl}(P_{Y|X},R)$ for
$R\in(0,R_{sl}(P_{Y|X})]$. Since $E(P_X,P_{Y|X},R)\leq
E(P_{Y|X},R)$, it follows from Theorem \ref{theorem1} that
\begin{eqnarray*}
E_{v}(P_{XY},R)\leq E_{sl}(P_{Y|X},H(P_X)-R)
\end{eqnarray*}
for $R\in[\max\{H(P_X)-R_{sl}(P_{Y|X}),0\}, H(P_X))$.

Note that the straight-line exponent holds for arbitrary block
codes; one can obtain further improvement at high rates by
leveraging bounds tailored to constant composition codes. Let
$E^*_{ex}(Q_{X},P_{Y|X},0)$ be the concave upper envelope of
$E_{ex}(Q_{X},P_{Y|X},0)$ considered as a function of $Q_{X}$. In
view of \cite[Excercise 5.21]{CK81B}, we have
\begin{eqnarray*}
E(Q_X,P_{Y|X},R)\leq E^*_{ex}(Q_X,P_{Y|X},0)
\end{eqnarray*}
for any $Q_X\in\mathcal{P}(\mathcal{X})$ and $R>0$. Now it follows
from Theorem \ref{theorem1} that
\begin{eqnarray*}
E_{v}(P_{XY},R)\leq E^*_{ex}(P_X,P_{Y|X},0)
\end{eqnarray*}
for $R<H(P_X)$.

The following theorem provides the second order expansion of
$E_v(P_{XY},R)$ at the Slepian-Wolf limit.
\begin{theorem}
Assuming $R_{cr}(P_{X},P_{X|Y})<I(P_X,P_{Y|X})$ (see Proposition
\ref{prop:cr} for the necessary and sufficient condition), we have
\begin{eqnarray*}
\lim\limits_{r\downarrow
0}\frac{E_v(P_{XY},H(X|Y)+r)}{r^2}=\frac{1}{2}\left[\sum\limits_{x,y}P_{XY}(x,y)\tau^2(x,y)-\sum\limits_{x}P_X(x)\left(\sum\limits_{y}\tau(x,y)P_{Y|X}(y|x)\right)^2\right]^{-1}
\end{eqnarray*}
where $\tau(x,y)=\log P_Y(y)-\log P_{Y|X}(y|x)$.
\end{theorem}
Remark: If $R_{cr}(P_{X},P_{Y|X})=I(P_X,P_{Y|X})$, then we have
$E_{v,rc}(P_{XY},R)=R-H(X|Y)$ for $R\geq H(X|Y)$, which implies
\begin{eqnarray*}
\lim\limits_{r\downarrow
0}\frac{E_v(P_{XY},H(X|Y)+r)}{r^2}=\infty.
\end{eqnarray*}
It is also worth noting that the second order expansion of
$E_v(P_{XY},R)$ at the Slepian-Wolf limit yields the
redundancy-error tradeoff constant of variable-rate Slepian-Wolf
coding derived in \cite{HLYJC08}.

\begin{proof}
Since $R_{cr}(P_{X},P_{X|Y})<I(P_X;P_{Y|Y})$, it follows that
$H(X|Y)+r\in(H(X|Y),H(P_X)-R_{cr}(P_{X},P_{Y|X}))$ when $r$
$(r>0)$ is sufficiently close to zero. In this case, we have
\begin{eqnarray*}
\frac{E_v(P_{XY},H(X|Y)+r)}{r^2}&=&\frac{E_{sp}(P_X,P_{Y|X},I(P_X,P_{Y|X})-r)}{r^2}\\
&=&\min\limits_{Q_{Y|X}:I(P_X,Q_{Y|X})\leq
I(P_X,P_{Y|X})-r}\frac{D(Q_{Y|X}\|P_{Y|X}|P_X)}{r^2}\\
&=&\min\limits_{Q_{Y|X}:I(P_X,Q_{Y|X})=I(P_X,P_{Y|X})-r}\frac{D(Q_{Y|X}\|P_{Y|X}|P_X)}{r^2},
\end{eqnarray*}
where the last equality follows from the fact that
$E_{sp}(P_X,P_{Y|X},R)$ is a strictly decreasing convex function
of $R$ for $R\in(R^{\infty}_{sp}(P_X,P_{Y|X}),I(P_X,P_{Y|X})]$.

Let $\Delta(x,y)=Q_{Y|X}(y|x)-P_{Y|X}(y|x)$ for $x\in\mathcal{X}$,
$y\in\mathcal{Y}$. Let $\Delta(y)=\sum_xP_X(x)\Delta(x,y)$ for
$y\in\mathcal{Y}$. By the Taylor expansion,
\begin{eqnarray*}
I(P_X,Q_{Y|X})&=&\sum\limits_{x,y}P_X(x)(P_{Y|X}(y|x)+\Delta(x,y))\log(P_{Y|X}(y|x)+\Delta(x,y))\\
&&-\sum\limits_{y}(P_{Y}(y)+\Delta(y))\log(P_Y(y)+\Delta(y))\\
&=&\sum\limits_{x,y}P_X(x)(P_{Y|X}(y|x)+\Delta(x,y))\left(\log P_{Y|X}(y|x)+\frac{\Delta(x,y)}{P_{Y|X}(y|x)}+o(\Delta(x,y))\right)\\
&&-\sum\limits_{y}(P_{Y}(y)+\Delta(y))\left(\log
P_Y(y)+\frac{\Delta(y)}{P_Y(y)}+o(\Delta(y))\right)\\
&=&I(P_X,P_{Y|X})-\sum\limits_{y}(\Delta(y)+\Delta(y)\log P_Y(y)+o(\Delta_y))\\
&&+\sum\limits_{x,y}P_X(x)(\Delta(x,y)+\Delta(x,y)\log
P_{Y|X}(y|x)+o(\Delta(x,y)))
\end{eqnarray*}
and
\begin{eqnarray*}
&&D(Q_{Y|X}\|P_{Y|X}|P_X)\\
&=&\sum\limits_{x,y}P_{X}(x)Q_{Y|X}(y|x)\log\frac{Q_{Y|X}(y|x)}{P_{Y|X}(y|x)}\\
&=&\sum\limits_{x,y}P_{X}(x)(P_{Y|X}(y|x)+\Delta(x,y))\log\left(1+\frac{\Delta(x,y)}{P_{Y|X}(y|x)}\right)\\
&=&\sum\limits_{x,y}P_{X}(x)(P_{Y|X}(y|x)+\Delta(x,y))\left(\frac{\Delta(x,y)}{P_{Y|X}(y|x)}-\frac{\Delta^2(x,y)}{2P^2_{Y|X}(y|x)}+o(\Delta^2(x,y))\right)\\
&=&\sum\limits_{x,y}P_{X}(x)\left(\frac{\Delta^2(x,y)}{2P_{Y|X}(y|x)}+o(\Delta^2(x,y))\right).
\end{eqnarray*}
Here $f(z)=o(z)$ means $\lim_{z\rightarrow 0}\frac{f(z)}{z}=0$.

As $r\downarrow 0$, we have $\Delta(y)\rightarrow 0$,
$\Delta(x,y)\rightarrow 0$ for all $x\in\mathcal{X},
y\in\mathcal{Y}$. Therefore, by ignoring the high order terms
which do not affect the limit, we get
\begin{eqnarray}
\lim\limits_{r\downarrow
0}\frac{E_v(P_{XY},H(X|Y)+r)}{r^2}=\lim\limits_{r\downarrow
0}\min\sum\limits_{x,y}\frac{P_X(x)\Delta^2(x,y)}{2P_{Y|X}(y|x)r^2}\label{minimization}
\end{eqnarray}
where the minimization is over $\Delta(x,y)$ ($x\in\mathcal{X},
y\in\mathcal{Y}$) subject to the constraints
\begin{enumerate}
\item $\sum_y\Delta(x,y)=0$ for all $x\in\mathcal{X}$;

\item $\sum_{x,y}P_X(x)\tau(x,y)\Delta(x,y)=r$.
\end{enumerate}
Introduce the Lagrange multipliers $\alpha(x)$
$(x\in\mathcal{X})$, $\beta$ for these constraints, and define
\begin{eqnarray*}
G=\sum\limits_{x,y}\frac{P_X(x)\Delta^2(x,y)}{2P_{Y|X}(y|x)}-\sum\limits_{x,y}\alpha(x)\Delta(x,y)-\beta\sum_{x,y}P_X(x)\tau(x,y)\Delta(x,y).
\end{eqnarray*}
The Karush-Kuhn-Tucker conditions yield
\begin{eqnarray*}
\frac{\partial G}{\partial\Delta(x,y)}=-\alpha(x)-\beta
P_X(x)\tau(x,y)+\frac{P_X(x)\Delta(x,y)}{P_{Y|X}(y|x)}=0,\quad
x\in\mathcal{X}, y\in\mathcal{Y}.
\end{eqnarray*}
Therefore, we have
\begin{eqnarray}
\Delta(x,y)=\beta\tau(x,y)P_{Y|X}(y|x)+\frac{P_{Y|X}(y|x)}{P_X(x)}\alpha(x).
\label{deltaxy}
\end{eqnarray}
Substituting (\ref{deltaxy}) into constraint 1), we obtain
\begin{eqnarray*}
\alpha(x)=-\beta P_X(x)\sum\limits_{y}\tau(x,y)P_{Y|X}(y|x)
\end{eqnarray*}
which, together with (\ref{deltaxy}), yields
\begin{eqnarray}
\Delta(x,y)=\beta\tau(x,y)P_{Y|X}(y|x)-\beta
P_{Y|X}(y|x)\sum\limits_{y'}\tau(x,y')P_{Y|X}(y'|x).\label{deltaxybeta}
\end{eqnarray}
Therefore, we have
\begin{eqnarray}
&&\sum\limits_{x,y}\frac{P_X(x)\Delta^2_{xy}}{2P_{Y|X}(y|x)}\nonumber\\
&&=\frac{\beta^2}{2}\sum\limits_{x,y}P_{XY}(x,y)\left[\tau(x,y)-\sum\limits_{y'}\tau(x,y')P_{Y|X}(y'|x)\right]^2\nonumber\\
&&=\frac{\beta^2}{2}\sum\limits_{x,y}P_{XY}(x,y)\left[\tau^2(x,y)-2\tau(x,y)\sum\limits_{y'}\tau(x,y')P_{Y|X}(y'|x)+\left(\sum\limits_{y'}\tau(x,y')P_{Y|X}(y'|x)\right)^2\right]\nonumber\\
&&=\frac{\beta^2}{2}\left[\sum\limits_{x,y}P_{XY}(x,y)\tau^2(x,y)-\sum\limits_{x}P_X(x)\left(\sum\limits_{y}\tau(x,y)P_{Y|X}(y|x)\right)^2\right].\label{numerator}
\end{eqnarray}
Constraint 2) and (\ref{deltaxybeta}) together yield
\begin{eqnarray}
\frac{r^2}{\beta^2}&=&\frac{1}{\beta^2}\left(\sum\limits_{x,y}P_X(x)\tau(x,y)\Delta(x,y)\right)^2\nonumber\\
&=&\left[\sum\limits_{x,y}P_X(x)\tau(x,y)\left(\tau(x,y)P_{Y|X}(y|x)-P_{Y|X}(y|x)\sum\limits_{y'}\tau(x,y')P_{Y|X}(y'|x)\right)\right]^2\nonumber\\
&=&\left[\sum\limits_{x,y}P_{XY}(x,y)\tau^2(x,y)-\sum\limits_{x}P_X(x)\left(\sum\limits_{y}\tau(x,y)P_{Y|X}(y|x)\right)^2\right]^2.\nonumber\\\label{denominator}
\end{eqnarray}
The proof is complete by substituting (\ref{numerator}) and
(\ref{denominator}) back into (\ref{minimization}).
\end{proof}

\section{Variable-Rate Slepian-Wolf Coding: Below the Slepian-Wolf Limit}\label{below}

\begin{definition}\label{def:varcorrectexp}
Given a joint probability distribution $P_{XY}$, we say a correct
decoding exponent $E^c\geq 0$ is achievable with variable-rate
Slepian-Wolf codes at rate $R$ if for any $\delta>0$, there exists
a sequence of variable-rate Slepian-Wolf codes $\{\varphi_n\}$
such that
\begin{eqnarray*}
&&\limsup\limits_{n\rightarrow\infty}R(\varphi_n,P_{XY})\leq R+\delta,\\
&&\liminf\limits_{n\rightarrow\infty}-\frac{1}{n}\log
P_c(\varphi_n,P_{XY})\leq E^c+\delta.
\end{eqnarray*}
The smallest achievable correct decoding exponent at rate $R$ is
denoted by $E^c_v(P_{XY},R)$.
\end{definition}

In view of Theorem \ref{theorem1}, it is tempting to conjecture
that $E^c_v(P_{XY},R)=E^c(P_X,P_{Y|X},H(P_X)-R)$. It turns out
this is not true. We shall show that $E^c_v(P_{XY},R)=0$ for all
$R$. Actually we have a stronger result --- the correct decoding
probability of variable-rate Slepian-Wolf coding can be bounded
away from zero even when $R<H(X|Y)$. This is in sharp contrast
with fixed-rate Slepian-Wolf coding for which the correct decoding
probability decays to zero exponentially fast if the rate is below
the Slepian-Wolf limit. To make the statement more precise, we
need the following definition.

\begin{definition}\label{def:varcorrectpro}
Given a joint probability distribution $P_{XY}$, we say a correct
decoding probability $P_{c,v}(P_{XY},R)$ is achievable with
variable-rate Slepian-Wolf codes at rate $R$ if for any
$\delta>0$, there exists a sequence of variable-rate Slepian-Wolf
codes $\{\varphi_n\}$ such that
\begin{eqnarray*}
&&\limsup\limits_{n\rightarrow\infty}R(\varphi_n,P_{XY})\leq R+\delta,\\
&&\limsup\limits_{n\rightarrow\infty}P_c(\varphi_n,P_{XY})\geq
P_{c,v}(P_{XY},R)-\delta.
\end{eqnarray*}
The largest achievable correct decoding probability at rate $R$ is
denoted by $P^{\max}_{c,v}(P_{XY},R)$.
\end{definition}

\begin{theorem}
$P^{\max}_{c,v}(P_{XY},R)=\frac{R}{H(X|Y)}$ for $R\in(0, H(X|Y)]$.
\end{theorem}
Remark: It is obvious that $P^{\max}_{c,v}(P_{XY},R)=1$ for
$R>H(X|Y)$. Moreover, since $P^{\max}_{c,v}(P_{XY},R)$ is a
monotonically increasing function of $R$, it follows that
$P^{\max}_{c,v}(P_{XY},0)=0$.

\begin{proof}
The intuition underlying the proof is as follows. Assume the rate
is below the Slepian-Wolf limit, i.e., $R<H(X|Y)$. For each type
$P$ in the neighborhood of $P_{X}$, the rate allocated to the type
class $\mathcal{T}_n(P)$ should be no less than $H(X|Y)$ in order
to correctly decode the sequences in $\mathcal{T}_n(P)$. However,
since almost all the probability are captured by the type classes
whose types are in the neighborhood of $P_X$, there is no enough
rate to protect all of them. Note that if the rate is evenly
allocated among these type classes, none of them can get enough
rate; consequently, the correct decoding probability goes to zero.
A good way is to protect only a portion of them to accumulate
enough rate. Specifically, we can protect $\frac{R}{H(X|Y)}$
fraction of these type classes so that the rate allocated to each
of them is about $H(X|Y)$ and leave the remaining type classes
unprotected. It turns out this strategy achieves the maximum
correct decoding probability as the block length $n$ goes to
infinity. Somewhat interestingly, although $E^c_v(P_{XY},R)\neq
E^c(P_X,P_{Y|X},H(P_X)-R)$, the function $E^c(P_X,P_{Y|X},\cdot)$
does play a fundamental role in establishing the correct result.

The proof is divided into two parts. Firstly, we shall show that
$P^{\max}_{c,v}(P_{XY},R)\geq\frac{R}{H(X|Y)}$. For any
$\epsilon>0$, define
\begin{eqnarray*}
\mathcal{U}(\epsilon)=\left\{P\in\mathcal{P}(\mathcal{X}):\|P-P_X\|\leq\epsilon\right\}.
\end{eqnarray*}
Since $P_X(x)>0$ for all $x\in\mathcal{X}$, we can choose
$\epsilon$ small enough so that
\begin{eqnarray*}
q_{\min}(\epsilon)\triangleq\min\limits_{P\in\mathcal{U}(\epsilon),x\in\mathcal{X}}P(x)>0.
\end{eqnarray*}
Using Stirling's approximation
\begin{eqnarray*}
\sqrt{2\pi
m}\left(\frac{m}{e}\right)^{m}e^{\frac{1}{12m+1}}<m!<\sqrt{2\pi
m}\left(\frac{m}{e}\right)^{m}e^{\frac{1}{12m}},
\end{eqnarray*}
we have, for any
$P\in\mathcal{U}(\epsilon)\cap\mathcal{P}_n(\mathcal{X})$,
\begin{eqnarray*}
\mbox{Pr}(X^n\in\mathcal{T}_n(P))&=&\frac{n!}{\prod_{x}
(nQ_{X}(x))!}\prod\limits_{x}\left[P_X(x)\right]^{nP(x)}\\
&\leq&\frac{\sqrt{2\pi n}e^{\frac{1}{12n}}}{\prod_{x}\sqrt{2\pi
nP(x)}}e^{-nD(P\|P_X)}\\
&\leq&\frac{\sqrt{2\pi}e^{\frac{1}{12n}}}{\prod_{x}\sqrt{2\pi
P(x)}}n^{-\frac{|\mathcal{X}|-1}{2}}\\
&\leq&\frac{\sqrt{2\pi}e^{\frac{1}{12n}}}{\prod_{x}\sqrt{2\pi
q_{\min}(\epsilon)}}n^{-\frac{|\mathcal{X}|-1}{2}},
\end{eqnarray*}
which implies that $\mbox{Pr}(X^n\in\mathcal{T}_n(P))$ converges
uniformly to zero as $n\rightarrow\infty$ for all
$P\in\mathcal{U}(\epsilon)\cap\mathcal{P}_n(\mathcal{X})$.
Moreover, it follows from the weak law of large numbers that
\begin{eqnarray*}
\lim\limits_{n\rightarrow\infty}\sum\limits_{P\in\mathcal{U}(\epsilon)\cap\mathcal{P}_n(\mathcal{X})}\mbox{Pr}(X^n\in\mathcal{T}_n(P))=1
\end{eqnarray*}
Therefore, for any $\delta>0$, $R\in(0,H(X|Y)]$, and sufficiently
large $n$, we can find a set
$\mathcal{S}_n\subseteq\mathcal{U}(\epsilon)\cap\mathcal{P}_n(\mathcal{X})$
such that
\begin{eqnarray*}
\frac{R}{H(X|Y)}-\delta\leq\sum\limits_{P\in\mathcal{S}_n}\mbox{Pr}(X^n\in\mathcal{T}_n(P))\leq\frac{R}{H(X|Y)}.
\end{eqnarray*}

Now consider a sequence of variable-rate Slepian-Wolf codes
$\{\varphi_n(\cdot)\}$ specified as follows.
\begin{enumerate}
\item The encoder sends of type of $X^n$ to the decoder, where
each type is uniquely represented by a binary sequence of length
$\lceil\log_2|\mathcal{P}_n(\mathcal{X})|\rceil$.

\item For each
$P\in\mathcal{S}_n$, the encoder partitions the type class
$\mathcal{T}_n(P)$ into $L_n$ subsets
$\mathcal{T}_n(P,1),\mathcal{T}_n(P,2),\cdots,\mathcal{T}_n(P,L_n)$.
If $X^n\in\mathcal{T}_n(P)$ for some $P\in\mathcal{S}_n$, the
encoder finds the subset $\mathcal{T}_n(P,i^*)$ that contains
$X^n$ and sends the index $i^*$ to the decoder, where each index
in $\{1,2,\cdots,L_n\}$ is uniquely represented by a binary
sequence of length $\lceil\log_2|L_n|\rceil$.

\item The remaining type classes are left uncoded.
\end{enumerate}

Specifically, we let
\begin{eqnarray*}
L_n=\left\lceil\left(2(n+1)^{|\mathcal{X}|^2}e^{n(H(X|Y)+\delta)}\right)\right\rceil.
\end{eqnarray*}
It follows from \cite[Theorem 2]{CK81} that for each
$P\in\mathcal{S}_n$, it is possible to partition the type class
$\mathcal{T}_n(P)$ into $L_n$ disjoint subsets
$\mathcal{T}_n(P,1),\mathcal{T}_n(P,2),\cdots,\mathcal{T}_n(P,L_n)$
so that
\begin{eqnarray*}
-\frac{1}{n}\log\mbox{Pr}(\widehat{X}^n\neq
X^n|X^n\in\mathcal{T}_n(P))\geq\min\limits_{Q_X\in\mathcal{U}(\epsilon)}\left[E_{rc}(Q_X,P_{Y|X},H(Q_X)-H(X|Y)-\delta)-\epsilon\right]
\end{eqnarray*}
uniformly for all $P\in\mathcal{S}_n$ when $n$ is sufficiently
large. In view of the fact that
$E_{rc}(P_X,P_{Y|X},I(P_X,P_{Y|X})-\delta)>0$ and that
$E_{rc}(Q_X,P_{Y|X},R)$ as a function of the pair $(Q_X,R)$ is
uniformly equicontinuous, we have
\begin{eqnarray*}
\min\limits_{Q_X\in\mathcal{U}(\epsilon)}\left[E_{rc}(Q_X,P_{Y|X},H(Q_X)-H(X|Y)-\delta)-\epsilon\right]\triangleq\kappa_1>0
\end{eqnarray*}
for sufficiently small $\epsilon$.

For this sequence of constructed variable-rate Slepian-wolf codes
$\{\varphi_n(\cdot)\}$, it can be readily verified that
\begin{eqnarray*}
\limsup\limits_{n\rightarrow\infty}R(\varphi_n,P_{XY})&=&\limsup\limits_{n\rightarrow\infty}\frac{1}{n\log_2e}\left[\lceil\log_2|\mathcal{P}_n(\mathcal{X})|\rceil+\sum\limits_{P\in\mathcal{S}_n}\mbox{Pr}\{X^n\in\mathcal{T}_n(P)\}L_n\right]\\
&\leq&\frac{R}{H(X|Y)}(H(X|Y)+\delta)
\end{eqnarray*}
and
\begin{eqnarray*}
\limsup\limits_{n\rightarrow\infty}P_c(\varphi_n,R)&\geq&\limsup\limits_{n\rightarrow\infty}\sum\limits_{P\in\mathcal{S}_n}\mbox{Pr}\{X^n\in\mathcal{T}_n(P)\}\left[1-\mbox{Pr}\{\widehat{X}^n\neq
X^n|X^n\in\mathcal{T}_n(P)\}\right]\\
&\geq&\limsup\limits_{n\rightarrow\infty}\sum\limits_{P\in\mathcal{S}_n}\mbox{Pr}\{(X^n\in\mathcal{T}_n(P)\}\left(1-e^{-n\kappa_1}\right)\\
&\geq&\frac{R}{H(X|Y)}-\delta.
\end{eqnarray*}
Since $\delta>0$ is arbitrary, it follows from Definition
\ref{def:varcorrectpro} that
$P^{\max}_{c,v}(P_{XY},R)\geq\frac{R}{H(X|Y)}$.

Now we proceed to prove the other direction. It follows from
Definition \ref{def:varcorrectpro} that for any $\delta>0$, there
exists a sequence of variable-rate Slepian-Wolf codes
$\{\varphi_n(\cdot)\}$ with
\begin{eqnarray*}
&&\limsup\limits_{n\rightarrow\infty}R(\varphi_n,P_{XY})\leq R+\delta,\\
&&\limsup\limits_{n\rightarrow\infty}P_c(\varphi_n,P_{XY})\geq
P^{\max}_{c,v}(R)-\delta.
\end{eqnarray*}
Define
\begin{eqnarray*}
r(\mathcal{T}_n(P))=\frac{1}{|\mathcal{T}_n(P)|n\log_2e}\sum\limits_{x^n\in\mathcal{T}_n(P)}l(\varphi_n(x^n)),\quad
P\in\mathcal{P}_n(\mathcal{X}).
\end{eqnarray*}
Since
$R(\varphi_n,P_{XY})=\sum_{P\in\mathcal{P}_n(\mathcal{X})}\mbox{Pr}\{X^n\in\mathcal{T}_n(P)\}r(\mathcal{T}_n(P))$,
we can interpret $r(\mathcal{T}_n(P))$ as the rate allocated to
the type class $\mathcal{T}_n(P)$.

For each
$P\in\mathcal{U}(\epsilon)\cap\mathcal{P}_n(\mathcal{X})$, suppose
$\varphi_n(\cdot)$ partitions the type class $\mathcal{T}_n(P)$
into $N(P)$ disjoint subsets
$\mathcal{T}_n(P,1),\cdots,\mathcal{T}_n(P,N(P))$ (i.e.,
$\varphi_n(x^n)=\varphi_n(\widetilde{x}^n)$ if
$x^n,\widetilde{x}^n\in\mathcal{T}_n(P,i)$ for some $i$, and
$\varphi_n(x^n)\neq\varphi_n(\widetilde{x}^n)$ if
$x^n\in\mathcal{T}_n(P,i),\widetilde{x}^n\in\mathcal{T}_n(P,j)$
for $i\neq j$). Define
\begin{eqnarray*}
&&\mathcal{I}_n(P,\delta)=\left\{i:\frac{1}{n}\log\frac{|\mathcal{T}_n(P)|}{|\mathcal{T}_n(P,i)|}\leq
H(X|Y)-\delta, i=1,2,\cdots,N(P) \right\},\\
&&\mathcal{I}^c_n(P,\delta)=\left\{i:\frac{1}{n}\log\frac{|\mathcal{T}_n(P)|}{|\mathcal{T}_n(P,i)|}>
H(X|Y)-\delta, i=1,2,\cdots,N(P) \right\}.
\end{eqnarray*}
Note that
\begin{eqnarray*}
r(\mathcal{T}_n(P))&\geq&\frac{1}{n}\sum\limits_{i=1}^{N(P)}\frac{|\mathcal{T}_n(P,i)|}{|\mathcal{T}_n(P)|}\log\frac{|\mathcal{T}_n(P)|}{|\mathcal{T}_n(P,i)|}\\
&\geq&\frac{1}{n}\sum\limits_{i\in\mathcal{I}^c_n(P,\delta)}\frac{|\mathcal{T}_n(P,i)|}{|\mathcal{T}_n(P)|}\log\frac{|\mathcal{T}_n(P)|}{|\mathcal{T}_n(P,i)|}\\
&\geq&(H(X|Y)-\delta)\sum\limits_{i\in\mathcal{I}^c_n(P,\delta)}\frac{|\mathcal{T}_n(P,i)|}{|\mathcal{T}_n(P)|},
\end{eqnarray*}
which implies
\begin{eqnarray*}
\sum\limits_{i\in\mathcal{I}_n(P,\delta)}\frac{|\mathcal{T}_n(P,i)|}{|\mathcal{T}_n(P)|}\geq
1-\frac{r(\mathcal{T}_n(P))}{H(X|Y)-\delta}.
\end{eqnarray*}

Each $\mathcal{T}_n(P,i)$ can be viewed as a constant composition
code of type $P$ and we have
\begin{eqnarray*}
\mbox{Pr}\{\widehat{X}^n=X^n|X^n\in\mathcal{T}_n(P,i)\}=
P_c(\mathcal{T}_n(P,i),P_{Y|X}).
\end{eqnarray*}
Note that for
$P\in\mathcal{U}(\epsilon)\cap\mathcal{P}_n(\mathcal{X})$ and
$i\in\mathcal{I}_n(P,\delta)$,
\begin{eqnarray*}
\frac{1}{n}\log|\mathcal{T}_n(P,i)|&\geq&\frac{1}{n}\log|\mathcal{T}_n(P)|-H(X|Y)+\delta\\
&\geq& H(P)-H(X|Y)+\delta-|\mathcal{X}|\frac{\log(n+1)}{n}.
\end{eqnarray*}
Therefore, it follows from \cite[Lemma 5]{DK79} that
\begin{eqnarray*}
-\frac{1}{n}\log
P_c(\mathcal{T}_n(P,i),P_{Y|X})\geq\min\limits_{Q_X\in\mathcal{U}(\epsilon)}E^c(Q_X,P_{Y|X},H(Q_X)-H(X|Y)+\delta-\epsilon)-\epsilon
\end{eqnarray*}
uniformly for all
$P\in\mathcal{U}(\epsilon)\cap\mathcal{P}_n(\mathcal{X})$ and
$i\in\mathcal{I}_n(P,\delta)$ when $n$ is sufficiently large. In
view of the fact that $E^c(P_X,P_{Y|X},I(P_X,P_{Y|X})+\delta)>0$
and that $E^c(Q_X,P_{Y|X},R)$ as a function of the pair $(Q_X,R)$
is uniformly equicontinuous, we have
\begin{eqnarray*}
\min\limits_{Q_X\in\mathcal{U}(\epsilon)}\left[E^c(Q_X,P_{Y|X},H(Q_X)-H(X|Y)+\delta-\epsilon)-\epsilon\right]\triangleq\kappa_2>0
\end{eqnarray*}
for sufficiently small $\epsilon$.

Now it is easy to see that
\begin{eqnarray*}
&&\liminf\limits_{n\rightarrow\infty}P_e(\varphi_n,P_{XY})\\
&\geq&\liminf\limits_{n\rightarrow\infty}\sum\limits_{P\in\mathcal{U}(\epsilon)\cap\mathcal{P}_n(\mathcal{X})}\mbox{Pr}\{X^n\in\mathcal{T}_n(P)\}\sum\limits_{i\in\mathcal{I}_n(P,\delta)}\frac{|\mathcal{T}_n(P,i)|}{|\mathcal{T}_n(P)|}\left(1-\mbox{Pr}\{\widehat{X}^n=X^n|X^n\in\mathcal{T}_n(P,i)\}\right)\\
&\geq&\liminf\limits_{n\rightarrow\infty}\sum\limits_{P\in\mathcal{U}(\epsilon)\cap\mathcal{P}_n(\mathcal{X})}\mbox{Pr}\{X^n\in\mathcal{T}_n(P)\}\sum\limits_{i\in\mathcal{I}_n(P,\delta)}\frac{|\mathcal{T}_n(P,i)|}{|\mathcal{T}_n(P)|}(1-e^{-n\kappa_2})\\
&\geq&\liminf\limits_{n\rightarrow\infty}\sum\limits_{P\in\mathcal{U}(\epsilon)\cap\mathcal{P}_n(\mathcal{X})}\mbox{Pr}\{X^n\in\mathcal{T}_n(P)\}\left(1-\frac{r(\mathcal{T}_n(P))}{H(X|Y)-\delta}\right)(1-e^{-n\kappa_2})\\
&\geq&\liminf\limits_{n\rightarrow\infty}\sum\limits_{P\in\mathcal{U}(\epsilon)\cap\mathcal{P}_n(\mathcal{X})}\mbox{Pr}\{X^n\in\mathcal{T}_n(P)\}(1-e^{-n\kappa_2})\\
&&-\limsup\limits_{n\rightarrow\infty}\sum\limits_{P\in\mathcal{U}(\epsilon)\cap\mathcal{P}_n(\mathcal{X})}\mbox{Pr}\{X^n\in\mathcal{T}_n(P)\}\frac{r(\mathcal{T}_n(P))}{H(X|Y)-\delta}(1-e^{-n\kappa_2})\\
&\geq&\liminf\limits_{n\rightarrow\infty}\sum\limits_{P\in\mathcal{U}(\epsilon)\cap\mathcal{P}_n(\mathcal{X})}\mbox{Pr}\{X^n\in\mathcal{T}_n(P)\}(1-e^{-n\kappa_2})\\
&&-\limsup\limits_{n\rightarrow\infty}\sum\limits_{P\in\mathcal{P}_n(\mathcal{X})}\mbox{Pr}\{X^n\in\mathcal{T}_n(P)\}\frac{r(\mathcal{T}_n(P))}{H(X|Y)-\delta}(1-e^{-n\kappa_2})\\
&=&\liminf\limits_{n\rightarrow\infty}\sum\limits_{P\in\mathcal{U}(\epsilon)\cap\mathcal{P}_n(\mathcal{X})}\mbox{Pr}\{X^n\in\mathcal{T}_n(P)\}(1-e^{-n\kappa_2})\\
&&-\limsup\limits_{n\rightarrow\infty}\frac{R(\varphi_n,P_{XY})}{H(X|Y)-\delta}(1-e^{-n\kappa_2})\\
&=&1-\frac{R+\delta}{H(X|Y)-\delta},
\end{eqnarray*}
which implies
\begin{eqnarray*}
\limsup\limits_{n\rightarrow\infty}P_c(\varphi_n,P_{XY})\leq\frac{R+\delta}{H(X|Y)-\delta}.
\end{eqnarray*}
Therefore, we have
\begin{eqnarray*}
P^{\max}_{c,v}(P_{XY},R)-\delta\leq\frac{R+\delta}{H(X|Y)-\delta}.
\end{eqnarray*}
Since $\delta>0$ is arbitrary, this completes the proof.
\end{proof}

\section{Example}\label{example}

Consider the joint distribution $P_{XY}$ over $\mathbb{Z}_2\times
\mathbb{Z}_2$ with $P_{X|Y}(1|0)=P_{X|Y}(0|1)=p$ and
$P_Y(0)=\tau$. We assume $p\in(0,\frac{1}{2})$,
$\tau\in(0,\frac{1}{2}]$. It is easy to compute that
\begin{eqnarray*}
&&P_{X}(0)=1-P_{X}(1)=\tau(1-p)+(1-\tau)p,\\
&&P_{Y|X}(1|0)=1-P_{Y|X}(0|0)=\frac{(1-\tau)p}{\tau(1-p)+(1-\tau)p},\\
&&P_{Y|X}(0|1)=1-P_{Y|X}(1|1)=\frac{\tau p}{\tau
p+(1-\tau)(1-p))}.
\end{eqnarray*}

For this joint distribution, we have $H(X|Y)=H_b(p)$, where
$H_b(\cdot)$ is the binary entropy function (i.e., $H_b(p)=-p\log
p-(1-p)\log(1-p)$). Given $R\in[0,\log 2]$, let $q$ be the unique
number satisfying $H_b(q)=R$ and $q\leq\frac{1}{2}$. It can be
verified that
\begin{eqnarray*}
&&E_{f,sp}(P_{XY},R)=D(q\|p),\quad R\in[H_b(p),\log 2],\\
&&E^c_f(P_{XY},R)=D(q\|p),\quad R\in[0,H_b(p)].
\end{eqnarray*}
Note that
\begin{eqnarray*}
E_{ex}(Q_{X},P_{Y|X},0)&=&-\sum\limits_{x,x'}Q_{X}(X)Q_{X}(x')\log\left[\sum\limits_{y}\sqrt{P_{Y|X}(y|x)P_{Y|X}(y|x')}\right]\\
&=&-2Q_{X}(0)Q_{X}(1)\log\left[\sum\limits_{y}\sqrt{P_{Y|X}(y|0)P_{Y|X}(y|1)}\right]
\end{eqnarray*}
which is a concave function of $Q_{X}$. Therefore,
\begin{eqnarray*}
E^*_{ex}(P_X,P_{Y|X},0)=E_{ex}(P_{X},P_{Y|X},0).
\end{eqnarray*}
Moreover, we have
\begin{eqnarray*}
E_{ex}(P_{Y|X},0)&=&\max\limits_{Q_{X}}E_{ex}(Q_{X},P_{Y|X},0)\\
&=&-\frac{1}{2}\log\left[\sum\limits_{y}\sqrt{P_{Y|X}(y|0)P_{Y|X}(y|1)}\right].
\end{eqnarray*}
It is easy to show that
\begin{eqnarray*}
E_{v,sp}(P_{XY},H(P_X))&=&E_{sp}(P_{X},P_{X|Y},0)\\
&=&\min\limits_{Q_{Y}}\sum\limits_{x}P_{X}(x)\sum\limits_{y}Q_{Y}(y)\log\frac{Q_{Y}(y)}{P_{Y|X}(y|x)}
\end{eqnarray*}
where the minimizer $Q^*_{Y}$ is given by
\begin{eqnarray*}
Q^*_{Y}(y)=\frac{\prod_xP_{Y|X}(y|x)^{P_{X}(x)}}{\sum_{y'}\prod_xP_{Y|X}(y'|x)^{P_{X}(x)}},\quad
y\in\mathcal{Y}.
\end{eqnarray*}
Define
\begin{eqnarray*}
&&E_{f,er}(P_{XY},R)=\max\{E_{f,ex}(P_{XY},R),
E_{f,rc}(P_{XY},R)\},\\
&&E_{v,er}(P_{XY},R)=\max\{E_{v,ex}(P_{XY},R),
E_{v,rc}(P_{XY},R)\}.
\end{eqnarray*}
We have
\begin{eqnarray*}
&&E_{f}(P_{XY},R)\geq E_{f,er}(P_{XY},R),\\
&&E_{v}(P_{XY},R)\geq E_{v,er}(P_{XY},R).
\end{eqnarray*}

\begin{figure}
\centering
\includegraphics[scale=0.9]{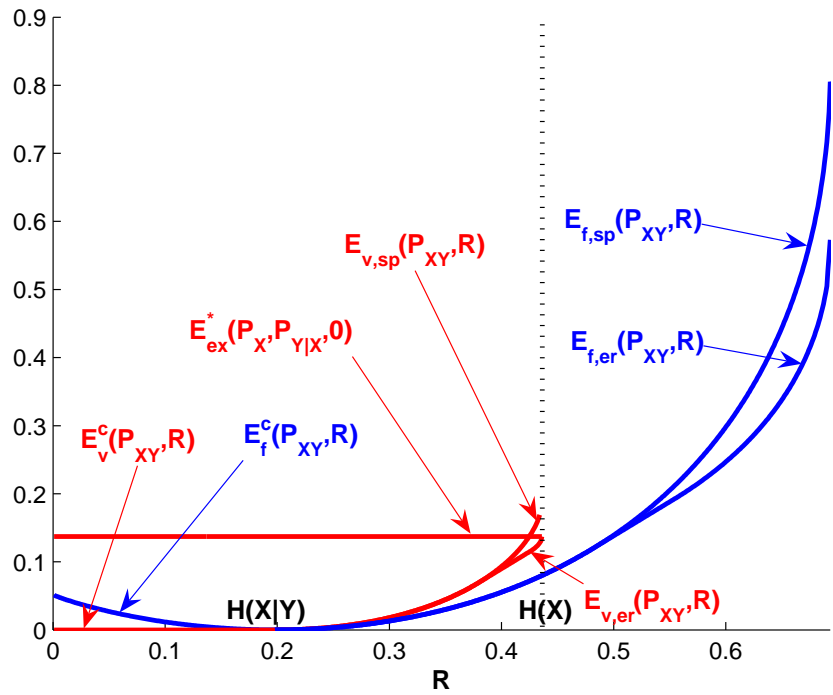}
\caption{$p=0.05$, $\tau=0.12$}\label{tau12}
\end{figure}

\begin{figure}
\centering
\includegraphics[scale=0.9]{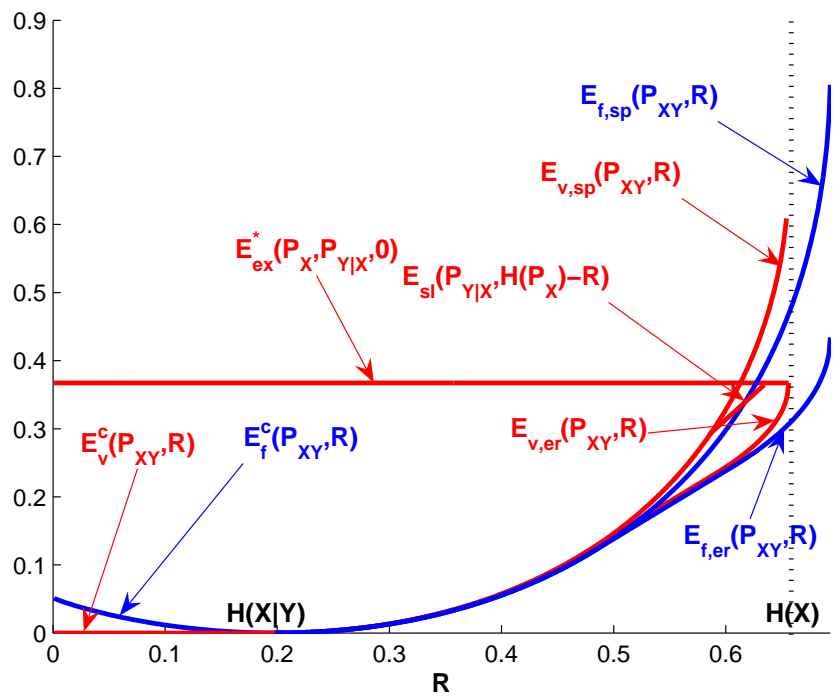}
\caption{$p=0.05$, $\tau=0.35$}\label{tau35}
\end{figure}

\begin{figure}
\centering
\includegraphics[scale=0.9]{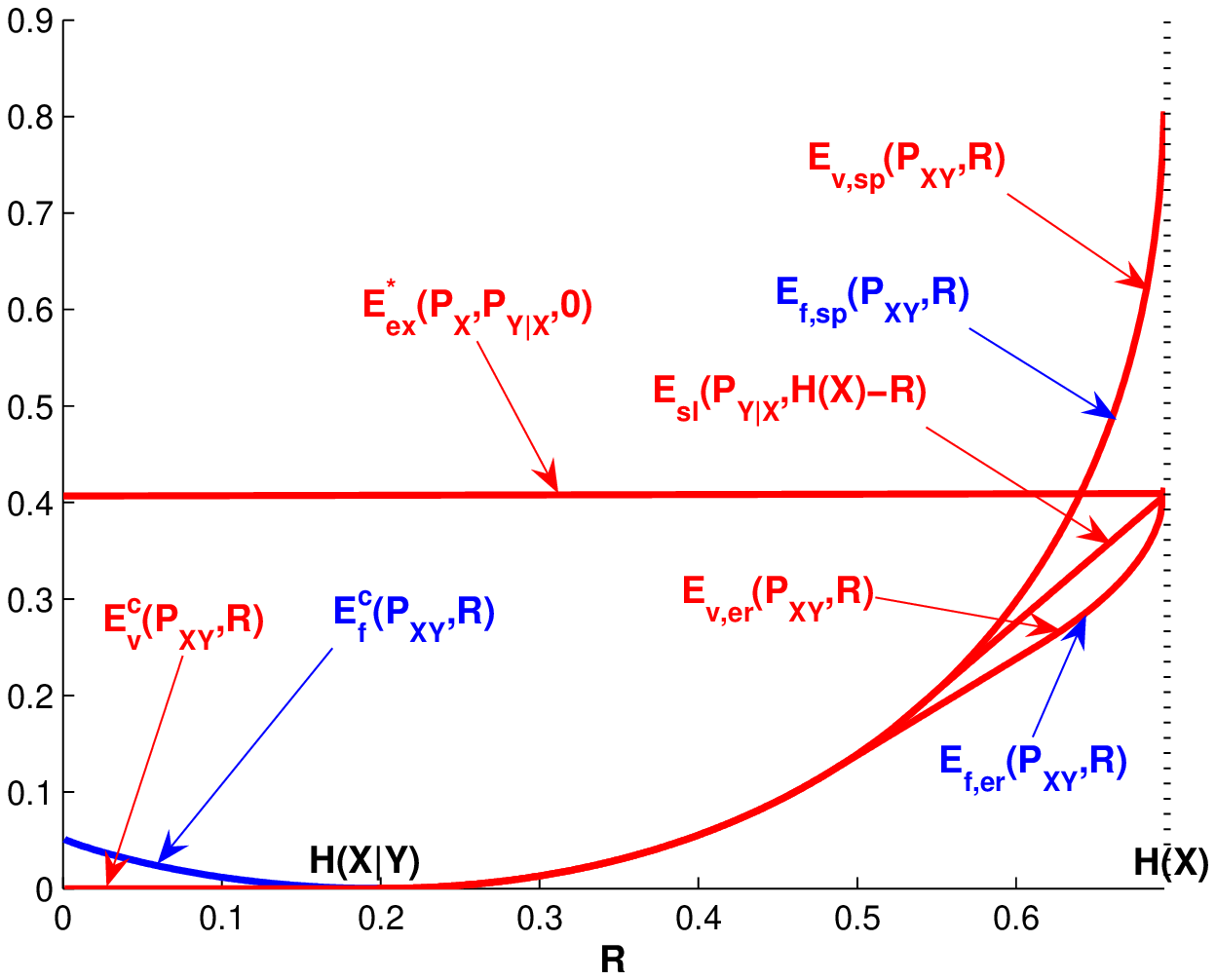}
\caption{$p=0.05$, $\tau=0.50$}\label{tau50}
\end{figure}

It can be seen from Fig. \ref{tau12} that the achievable error
exponent $E_{v,er}(P_{XY},R)$ of variable-rate Slepian-Wolf coding
can completely dominate the sphere packing exponent
$E_{f,sp}(P_{XY},R)$ of fixed-rate Slepian-Wolf coding. The gain
of variable-rate coding gradually diminishes as
$\tau\rightarrow\frac{1}{2}$ (see Fig. \ref{tau35} and Fig.
\ref{tau50}).

\section{Concluding Remarks}\label{conclusion}

We have studied the reliability function of variable-rate
Slepian-Wolf coding. An intimate connection between variable-rate
Slepian-Wolf codes and constant composition codes has been
revealed. It is shown that variable-rate Slepian-Wolf coding can
outperform fixed-rate Slepian-Wolf coding in terms of rate-error
tradeoff. Finally, we would like to mention that Theorem \ref{theorem1} has been generalized by Weinberger and Merhav in their recent paper on
the optimal tradeoff between the error exponent and the excess-rate exponent of variable-rate Slepian-Wolf coding \cite{WM15}.

\appendices
\section{Proof of Proposition \ref{prop:cr}}\label{app:prop:cr}
In view of (\ref{constrc}) and (\ref{rcsp}), we have
$R_{cr}(Q_X,W_{Y|X})=I(Q_X,W_{Y|X})$ if and only if the minimum of
the convex optimization problem
\begin{eqnarray}
\min\limits_{V_{Y|X}}D(V_{Y|X}\|W_{Y|X}|Q_X)+I(Q_X,V_{Y|X})
\label{prop:minimum}
\end{eqnarray}
is achieved at $V_{Y|X}=W_{Y|X}$. Let $V^*_{Y|X}$ be a minimizer
to the above optimization problem. Note that for $x,y$ such that
$Q_X(x)W_{Y|X}(y|x)=0$, there is no loss of generality in setting
$V^*_{Y|X}(y|x)=W_{Y|X}(y|x)$. Let
$\mathcal{A}=\{x\in\mathcal{X}:Q_X(x)>0\}$ and
$\mathcal{B}_x=\{y\in\mathcal{Y}:W_{Y|X}(y|x)>0\}$ for
$x\in\mathcal{A}$. We can rewrite (\ref{prop:minimum}) in the
following equivalent form:
\begin{eqnarray*}
\min\limits_{V_{Y|X}(y|x):x\in\mathcal{A},y\in\mathcal{B}_x}\sum\limits_{x\in\mathcal{A},y\in\mathcal{B}_x}Q_X(x)V_{Y|X}(y|x)\log\frac{V^2_{Y|X}(y|x)}{W_{Y|X}(y|x)\sum_{x'\in\mathcal{A}}Q_X(x')V_{Y|X}(y|x')}
\end{eqnarray*}
subject to
\begin{eqnarray*}
&&V_{Y|X}(y|x)\geq 0\quad\mbox{for all }x\in\mathcal{A},
y\in\mathcal{B}_x,\\
&&\sum\limits_{y\in\mathcal{B}_x}V_{Y|X}(y|x)=1\quad\mbox{for all
}x\in\mathcal{A}.
\end{eqnarray*}
Define
\begin{eqnarray*}
G'&=&\sum\limits_{x\in\mathcal{A},y\in\mathcal{B}_x}Q_X(x)V_{Y|X}(y|x)\log\frac{V^2_{Y|X}(y|x)}{W_{Y|X}(y|x)\sum_{x'\in\mathcal{A}}Q_X(x')V_{Y|X}(y|x')}\\
&&-\sum\limits_{x\in\mathcal{A},y\in\mathcal{B}_x}\alpha(x,y)-\sum\limits_{x\in\mathcal{A},y\in\mathcal{B}_x}\beta(x)V_{Y|X}(y|x),
\end{eqnarray*}
where $\alpha(x,y)\in\mathbb{R}_+$
($x\in\mathcal{A},y\in\mathcal{B}_x$) and $\beta(x)\in\mathbb{R}$
($x\in\mathcal{A}$). The Karush-Kuhn-Tucker conditions yield
\begin{eqnarray*}
&&\left.\frac{\partial G'}{\partial
V_{Y|X}(y^*|x^*)}\right|_{V_{Y|X}(y^*|x^*)=V^*_{Y|X}(y^*|x^*)}=2Q_X(x^*)\log
V^*_{Y|X}(y^*|x^*)+Q_X(x^*)-Q_X(x^*)\log
W_{Y|X}(y^*|x^*)\\
&&\hspace{2.42in}-Q_X(x^*)\log\sum_{x'\in\mathcal{A}}Q_X(x')V^*_{Y|X}(y^*|x')-\alpha(x^*,y^*)-\beta(x^*)\\
&&\hspace{2.28in}=0\quad\mbox{for all
}x^*\in\mathcal{A},y^*\in\mathcal{B}_{x^*},\\
&&V^*_{Y|X}(y^*|x^*)\geq 0\quad\mbox{for all }x^*\in\mathcal{A},
y^*\in\mathcal{B}_{x^*},\\
&&\sum\limits_{y^*\in\mathcal{B}_{x^*}}V^*_{Y|X}(y^*|x^*)=1\quad\mbox{for
all
}x^*\in\mathcal{A},\\
&&\alpha(x^*,y^*)V^*_{Y|X}(y^*|x^*)=0\quad\mbox{for all
}x^*\in\mathcal{A},y^*\in\mathcal{B}_{x^*}.
\end{eqnarray*}
By the complementary slackness conditions (i.e.,
$V^*_{Y|X}(y^*|x^*)>0\Rightarrow\alpha(x^*,y^*)=0$), we have
$V^*_{Y|X}=W_{Y|X}$ if and only if for all $x^*\in\mathcal{A}$,
$y^*\in\mathcal{B}_{x^*}$,
\begin{eqnarray*}
Q_X(x^*)\log
W_{Y|X}(y^*|x^*)+Q_X(x^*)-Q_X(x^*)\log\sum_{x'\in\mathcal{A}}Q_X(x')W_{Y|X}(y^*|x')-\beta(x^*)=0,
\end{eqnarray*}
i.e., the value of
\begin{eqnarray*}
\log\frac{W_{Y|X}(y|x)}{\sum_{x'}Q_X(x')W_{Y|X}(y|x')}
\end{eqnarray*}
does not depend on $y$ for all $x, y$ such that
$Q_X(x)W_{Y|X}(y|x)>0$.

\section{Proof of Proposition \ref{prop:constantbounds}}\label{app:prop:constantbounds}
\begin{enumerate}
\item It is known \cite[Exercise 5.17]{CK81B} that for every $R>0$,
$\delta>0$ and every $P\in\mathcal{P}_n(\mathcal{X})$ there exists
a constant composition code
$\mathcal{C}_n\subseteq\mathcal{T}_n(P)$ such that
\begin{eqnarray*}
&&R(\mathcal{C}_n)\geq R-\delta,\\
&&-\frac{1}{n}\log P_{e,\max}(\mathcal{C}_n,W_{Y|X})\geq
E_{ex}(P,W_{Y|X},R)-\delta
\end{eqnarray*}
whenever $n\geq n_0(|\mathcal{X}|,|\mathcal{Y}|,\delta)$. Let
$P_n$ be a sequence of types with
$P_n\in\mathcal{P}_n(\mathcal{X})$ and
\begin{eqnarray*}
\lim\limits_{n\rightarrow\infty}\|P_n-Q_X\|=0.
\end{eqnarray*}
Define
\begin{eqnarray*}
V^*_n=\arg\min\left[\sum\limits_{x,\widetilde{x}}P_n(x)V_n(\widetilde{x}|x)d_{W_{Y|X}}(x,\widetilde{x})+I(P_n,V_n)-R\right]
\end{eqnarray*}
where the minimization is over
$V_n:\mathcal{X}\rightarrow\mathcal{X}$ subject to the constraints
\begin{eqnarray*}
&&\sum\limits_{x}P_n(x)V_n(\widetilde{x}|x)=P_n(\widetilde{x}),\quad\mbox{for
all }\widetilde{x}\in\mathcal{X},\\
&&I(P_n,V_n)\leq R.
\end{eqnarray*}
Note that $\{V^*_n\}$ must contain a converging subsequence
$\{V^*_{n_k}\}$. Define
\begin{eqnarray*}
V^*=\lim\limits_{n\rightarrow\infty}V^*_{n_k}.
\end{eqnarray*}
It is easy to verify that
\begin{eqnarray*}
\sum\limits_{x\in\mathcal{X}}Q_X(x)V^*(\widetilde{x}|x)&=&\lim\limits_{k\rightarrow\infty}\sum\limits_{x\in\mathcal{X}}P_{n_k}(x)V^*_{n_k}(\widetilde{x}|x)\\
&=&\lim\limits_{k\rightarrow\infty}P_n(\widetilde{x})\\
&=&Q_X(\widetilde{x}),\quad\mbox{for
all }\widetilde{x}\in\mathcal{X},\\
I(Q_X,V^*)&=&\lim\limits_{k\rightarrow\infty}I(P_{n_k},V^*_{n_k})\\
&\leq&R.
\end{eqnarray*}
Therefore, we have
\begin{eqnarray*}
&&\limsup\limits_{n\rightarrow\infty}E_{ex}(P_n,W_{Y|X},R)\\
&\geq&\limsup\limits_{k\rightarrow\infty}E_{ex}(P_{n_k},W_{Y|X},R)\\
&=&\limsup\limits_{k\rightarrow\infty}\sum\limits_{x,\widetilde{x}\in\mathcal{X}}P_{n_k}(x)V^*_{n_k}(\widetilde{x}|x)d_{W_{Y|X}}(x,\widetilde{x})+I(P_{n_k},V_{n_k})-R\\
&\geq&\sum\limits_{x,\widetilde{x}\in\mathcal{X}}Q_X(x)V^*(\widetilde{x}|x)d_{P_{Y|X}}(x,\widetilde{x})+I(Q_X,V^*)-R\\
&\geq&E_{ex}(Q_X,W_{Y|X},R).
\end{eqnarray*}

It is also known \cite[Theorem 5.2]{CK81B} that for every $R>0$,
$\delta>0$ and every $P\in\mathcal{P}_n(\mathcal{X})$ there exists
a constant composition code
$\mathcal{C}_n\subseteq\mathcal{T}_n(P)$ such that
\begin{eqnarray*}
&&R(\mathcal{C}_n)\geq R-\delta,\\
&&-\frac{1}{n}\log P_{e,\max}(\mathcal{C}_n,W_{Y|X})\geq
E_{rc}(P,W_{Y|X},R)-\delta
\end{eqnarray*}
whenever $n\geq n_0(|\mathcal{X}|,|\mathcal{Y}|,\delta)$. So it
can be readily shown that
\begin{eqnarray*}
E(Q_X,W_{Y|X},R)\geq E_{rc}(Q_X,W_{Y|X},R)
\end{eqnarray*}
by invoking the fact that $E_{rc}(P,W_{Y|X},R)$ as a function of
the pair $(P,R)$ is uniformly equicontinuous \cite[Lemma
5.5]{CK81B}. The proof is complete.

\item By Definition \ref{def:constantcomposition}, for every $R>0$,
$\delta>0$  there exists a sequence of block channel codes codes
$\{\mathcal{C}_n\}$ with
$\mathcal{C}_n\subseteq\mathcal{T}_n(P_n)$ for some
$P_n\in\mathcal{P}_n(\mathcal{X})$ such that
\begin{eqnarray}
&&\lim\limits_{n\rightarrow\infty}\|P_n-Q_X\|=0,\nonumber\\
&&\liminf\limits_{n\rightarrow\infty}R(\mathcal{C}_n)\geq R-\delta,\nonumber\\
&&\limsup\limits_{n\rightarrow\infty}-\frac{1}{n}\log
P_{e,\max}(\mathcal{C}_n,W_{Y|X})\geq
E(Q_X,W_{Y|X},R)-\delta.\label{prop1:reliability}
\end{eqnarray}
For simplicity, we assume $R(\mathcal{C}_n)\geq R-\delta$ for all
$n$. Now it follows from Theorem 5.3 in [8] that
\begin{eqnarray}
-\frac{1}{n}\log\left[2P_{e,\max}(\mathcal{C}_n,W_{Y|X})\right]\leq
E_{sp}(P_n,W_{Y|X},R-2\delta)(1+\delta)\label{prop1:spbound}
\end{eqnarray}
whenever $n\geq n_0(|\mathcal{X}|,|\mathcal{Y}|,\delta)$. Let
\begin{eqnarray*}
V^*_{Y|X}=\arg\min\limits_{V_{Y|X}:I(Q_X,V_{Y|X})\leq
R-3\delta}D(V_{Y|X}\|W_{Y|X}|Q_X).
\end{eqnarray*}
Without loss of generality, we can set
$V^*_{Y|X}(\cdot|x)=W_{Y|X}(\cdot|x)$ for all
$x\in\{x'\in\mathcal{X}:Q_X(x')=0\}$. It is easy to see that there
exists an $\epsilon>0$ such that
\begin{eqnarray*}
&&I(P,V^*_{Y|X})\leq R-2\delta,\\
&&D(V^*_{Y|X}\|W_{Y|X}|P)\leq D(V^*_{Y|X}\|W_{Y|X}|Q_X)+\delta
\end{eqnarray*}
for all $P\in\mathcal{P}(\mathcal{X})$ with
$\|P-Q_X\|\leq\epsilon$. Therefore, for all sufficiently large
$n$,
\begin{eqnarray}
E_{sp}(P_n,W_{Y|X},R-2\delta)&=&\min\limits_{V_{Y|X}:I(P_n,V_{Y|X})\leq
R-3\delta}D(V_{Y|X}\|W_{Y|X}|P_n)\nonumber\\
&\leq&D(V^*_{Y|X}\|W_{Y|X}|P_n)\nonumber\\
&\leq& D(V^*_{Y|X}\|W_{Y|X}|Q_X)+\delta\nonumber\\
&=& E_{sp}(Q_X,W_{Y|X},R-3\delta)+\delta.\label{prop1:sp}
\end{eqnarray}
Combining (\ref{prop1:reliability}), (\ref{prop1:spbound}) and
(\ref{prop1:sp}), we get
\begin{eqnarray*}
E(Q_X,W_{Y|X},R)-\delta\leq
[E_{sp}(Q_X,W_{Y|X},R-3\delta)+\delta](1+\delta).
\end{eqnarray*}
In view of the fact that $\delta>0$ is arbitrary and that for
fixed $P$ and $W_{Y|X}$, $E_{sp}(P,W_{Y|X},R)$ is a decreasing
continuous convex function of $R$ in the interval where it is
finite \cite[Lemma 5.4]{CK81B}, the proof is complete.

\item It is known \cite[Lemma 5]{DK79} that for every $R>0$,
$\delta>0$, every constant composition code $\mathcal{C}_n$ of
common type $P$ for some $P\in\mathcal{P}_n(\mathcal{X})$ and rate
$R(\mathcal{C}_n)\geq R+\delta$ has
\begin{eqnarray*}
-\frac{1}{n}\log P_c(\mathcal{C}_n,W_{Y|X})\geq
\min\limits_{V_{Y|X}}\left[D(V_{Y|X}\|W_{Y|X}|P)+|R-I(P,V_{Y|X})|^+\right]-\delta
\end{eqnarray*}
whenever $n\geq n_0(|\mathcal{X}|,|\mathcal{Y}|,\delta)$.
Moreover, it is also known \cite[Lemma 2]{DK79}\cite[Excercise
5.16]{CK81B} that for every $R>0$, $\delta>0$ and every
$P\in\mathcal{P}_n(\mathcal{X})$ there exists a constant
composition code $\mathcal{C}_n\subseteq\mathcal{T}_n(P)$ such
that
\begin{eqnarray*}
&&R(\mathcal{C}_n)\geq R-\delta,\\
&&-\frac{1}{n}\log P_c(\mathcal{C}_n,W_{Y|X})\geq
\min\limits_{V_{Y|X}}\left[D(V_{Y|X}\|W_{Y|X}|P)+|R-I(P,V_{Y|X})|^+\right]+\delta
\end{eqnarray*}
whenever $n\leq n_0(|\mathcal{X}|,|\mathcal{Y}|,\delta)$. In view
of the fact that
$\min_{V_{Y|X}}\left[D(V_{Y|X}\|W_{Y|X}|P)+|R-I(P,V_{Y|X})|^+\right]$
as a function of the pair $(P,R)$ is uniformly equicontinuous, it
can be readily shown that
\begin{eqnarray*}
E_c(Q_X,W_{Y|X},R)=\min\limits_{V_{Y|X}}\left[D(V_{Y|X}\|W_{Y|X}|Q_X)+|R-I(P,V_{Y|X})|^+\right].
\end{eqnarray*}
The proof is complete.
\end{enumerate}

\end{document}